\documentclass[11pt,a4paper,twoside]{article}

\usepackage{a4wide, amssymb, amsmath, amsthm, graphics, comment, xspace, enumerate}
\usepackage{graphicx,multirow,rotating}
\usepackage[a4paper,colorlinks=true,citecolor=blue,urlcolor=blue,linkcolor=blue,bookmarksopen=true,unicode=true,
          pdffitwindow=true]{hyperref}
\usepackage[section]{algorithm}
\usepackage[noend]{algpseudocode}
\usepackage{mathrsfs,epsfig,amsfonts,amscd,caption, tikz,breqn}
\usepackage{cite}
\usepackage{tabulary}
\usepackage{multirow,rotating}

\usepackage[titletoc,toc,title]{appendix}
\usepackage[T1]{fontenc}
\usepackage{alltt}
\usepackage{color}
\definecolor{string}{rgb}{0.7,0.0,0.0}
\definecolor{comment}{rgb}{0.13,0.54,0.13}
\definecolor{keyword}{rgb}{0.0,0.0,1.0}

\newtheorem{te}{Theorem}[section]

\newtheorem{conjecture}{Conjecture}[section]

\newcommand{\beq}{\begin{eqnarray}}
\newcommand{\eeq}{\end{eqnarray}}
\newcommand{\beqs}{\begin{eqnarray*}}
\newcommand{\eeqs}{\end{eqnarray*}}

\newcommand{\ds}{\displaystyle}
\allowdisplaybreaks

\newcommand{\irr}{{\rm irr}}
\newcommand{\imb}{{\rm imb}}
\newcommand{\Var}{{\rm Var}}
\newcommand{\CS}{{\rm CS}}
\newcommand{\tc}{{\rm t}}

\begin{document}
\title{{On the Irregularity of Some  Molecular Structures } }
\maketitle
\begin{center}
{\large \bf  Hosam Abdo$^a$, Darko Dimitrov$^b$, Wei Gao$^c$}
\end{center}
\baselineskip=0.20in
\begin{center}
{\it $^a$Institut f\"ur Informatik, Freie Universit\"{a}t Berlin,
\\ Takustra{\ss}e 9, D--14195 Berlin, Germany} \\E-mail: {\tt abdo@mi.fu-berlin.de} \\[2mm]
\end{center}

\baselineskip=0.20in
\begin{center}
{\it $^b$Hochschule f\"ur Technik und Wirtschaft Berlin,
\\ Wilhelminenhofstra{\ss}e 75A, D--12459 Berlin, Germany}
\\E-mail: {\tt darko.dimitrov11@gmail.com}
\end{center}

\baselineskip=0.20in
\begin{center}
{\it $^c$School of Information Science and Technology, Yunnan
Normal University
\\ Kunming 650500, China}
\\E-mail: {\tt gaowei@ynnu.edu.cn}
\end{center}

\vspace{6mm}
\begin{abstract}
Measures of the irregularity of chemical graphs could be helpful for
QSAR/QSPR studies and for the descriptive purposes of biological and
chemical properties, such as melting and boiling points, toxicity
and resistance.
Here we consider the following four established irregularity measures:
 the irregularity index  by Albertson, the total irregularity, the variance of vertex degrees and the Collatz-Sinogowitz index.
Through the means of graph structural analysis and derivation, we study the
above-mentioned irregularity measures of several chemical molecular graphs
which frequently appear in chemical, medical and material
engineering, as well as the nanotubes: $TUC_4 C_8(S)$, $TUC_4 C_8(R)$, Zig-Zag
$TUHC_{6}$, $TUC_4$, Armchair $TUVC_{6}$, then dendrimers $T_{k,d}$ and the circumcoronene series of benzenoid $H_k$. In
addition, the irregularities of Mycielski's constructions of  cycle
and path graphs are analyzed.
\end{abstract}
{\small \hspace{0.25cm} \textbf{Keywords:}
Irregularity indices, molecular structures, nanotube, dendrimer, circumcoronene of benzenoid 
%
%
\medskip
\section[Introduction]{Introduction}
Nowadays, due to the increasing need of engineering applications in the fields of
transportation, aerospace, military and other various industrial
fields, 
there has been an accelerating demand for high performance materials.
The deterioration of the global environment makes the
original virus mutate at a greater pace, causing new diseases to emerge,
which increase mankind's demand for new drugs. It is with the
continuous improvements on chemical technology
that the new materials and new drugs are
discovered. Each year, these ever-increasing supply of new drugs and materials
meets the human needs in the industrial and medical fields.
However, with the new chemical substances there is a real necessity for a lot of chemical experiments to
test their properties, which would require a lot of researchers,
material and financial resources. On the other hand, in Southeast
Asia, Latin America, Africa among other developing countries and
regions, their governments cannot invest enough money to organize
people, purchase equipment and reagents to detect the properties
of these new compounds, which is one of the main reasons why these
countries fall behind in the fundamental industrial and medical fields.
Fortunately, early studies have shown that properties of the
compound and its molecular structure are inextricably linked. By
studying the corresponding molecular structure of the material and
drug, we can understand the chemical and pharmacological properties
of the compound. This discovery makes theoretical chemistry
an important branch of chemistry that attracts more and more
attention.

In standard theoretical chemistry, the chemical molecular
structure is expressed as a graph: each vertex denotes an atom of a
molecule and each edge between the corresponding vertices
expresses covalent bounds between the atoms. This graph obtained from
a chemical molecular structure is often called the molecular
graph.
A topological chemical index defined on molecular graph $G$ can
be regarded as a real-valued function $f: G\to \Bbb R$ which
assigns each molecular structure to a real number. In the past
four decades, researchers in chemical and mathematical science
have introduced several important indices, such as the Zagreb index,
the PI index, the eccentric index, the atom-bond connectivity index, the forgotten index and
the Wiener index e.g, to predict the characteristics of drugs,
nanomaterials and other chemical compounds. There were several
articles contributing to manifest these topological indices of
special molecular structures in nanomaterials, chemical, biological
and pharmaceutical engineering and in extremal molecular structures \cite{adg-etrfi-2016,Gao1,Gao2, Gao3, Gao4}.

Let $G$ be a simple undirected graph with $|V(G)|=n$ vertices and $|E(G)|=m$ edges. The \emph{degree} of a vertex $v$ in $G$
is the number of edges incident with $v$ and it is denoted by $d_G(v)$. A graph $G$ is \emph{regular} if all its vertices have the same degree,
otherwise it is \emph{irregular}. In many applications and problems in chemistry and pharmacy, it is of great importance to know how irregular a given graph is.

There are many ways to define  a regularity of a graph.
Let ${\imb}(e)=\left|d_G(u)-d_G(v)\right|$ be the \emph{imbalance} of an edge $e=uv \in E$.
In \cite{AlbertsonIrr}, Albertson defined the \emph{irregularity} of $G$ as
\beq \label{eqn:003}
{\irr}(G) = \sum_{e \in E(G)}  {\imb}(e) = \sum_{uv \in E(G)} | d_G(u) - d_G(v) |.
\eeq

It is shown in \cite{AlbertsonIrr} that for a graph $G$,
${\irr}(G) < 4 n^3/27$ and that this bound can be approached
arbitrarily close. This bound was slightly improved in
\cite{Dimit-Abdo1}. Albertson also presented upper bounds on
irregularity for bipartite graphs, triangle-free graphs and a
sharp upper bound for trees. Some claims about bipartite graphs
given in Albertson \cite{AlbertsonIrr} have been formally proved
in Henning and Rautenbach \cite{HennRaut}. Related to
Albertson's work is the work of Hansen and M{\'e}lot \cite{HansenMelot},
who characterized the graphs with $n$ vertices and $m$ edges with
maximal irregularity.
%

%
In \cite{Dimit-Abdo}, a new measure of irregularity
of a graph, so-called the {\it total irregularity} of a graph, was
defined as
\beq \label{eqn:002}
{\irr_t}(G) = \frac{1}{2}\sum_{u, v\in V(G)} \left| d_G(u)-d_G(v) \right|.
\eeq

Moreover, in \cite{Dimit-Abdo}  a sharp upper bound
of the total irregularity ${\irr_t}$ was given and the graphs of
small and maximal total irregularity were characterized.
The comparison between the irregularity ${\irr}$ and the total irregularity
 ${\irr_t}$ of a graph was studied in \cite{Dar-Riste}.

Two other most frequently used graph topological indices that
measure how irregular a graph is, are the \emph{ variance of degrees}
and the \emph{ Collatz-Sinogowitz index} \cite{CollSin57}. 
 For graph $G$ 
let $\lambda_1$ be
the largest eigenvalue of the adjacency matrix $A =
(a_{ij})$ (with $a_{ij} = 1$ if vertices $i$ and $j$ are joined by
an edge and $0$ otherwise). 
A sequence of non-negative integers $d_1, ..., d_n$ is a {\em graphic sequence},
or a {\em degree sequence}, if there exists a graph $G$ with $V(G) = \{v_1, ..., v_n\}$ such that $d(v_i) = d_i$.
By $n_i$ we denote the number of
vertices of degree $i$ for $i =1, 2,\dots, n - 1$ and by $d_1, ...,
d_n$ the degree sequence of the graph $G$, where $n_i$ is the
number of vertices of degree $i$ for $i = 1, 2, \cdots, n-1$. The
variance $\Var({\it G})$ of the vertex degrees of the graph $G$ is
\beq \label{eqn:001_3} \Var({\it G}) &=& \frac{1}{n} \sum^n_{i=1}
d^2_i -\frac{1}{n^2} (\sum^n_{i=1} d_i)^2=\frac{1}{n}
\sum_{i=1}^{n-1} n_i \left( i - \frac{2 m}{n}\right)^2. \eeq
%

The graph $G = (V,E)$ of order $n=|V(G)|$, size $m = |E(G)|$, maximum degree $\Delta$ and a real
$(0,1)-$adjacency matrix $A(G)=(a_{ij})$, where $a_{ij} = 1$ if the vertices $i$ and $j$ are adjacent
otherwise $a_{ij} = 0$. Since $A$ is symmetric, its eigenvalues $\lambda_1, \cdots, \lambda_n$ are
real and we assume that $\lambda_1 \geq \lambda_2 \geq \dots\ \geq \lambda_n$. Accordingly we write
$\lambda_i (G) = \lambda_i (A) = \lambda_i$, $(i = 1, \cdots, n)$.
The eigenvalues $\lambda_1, \cdots, \lambda_n$ refers to the {\it spectrum} of $G$.
The largest eigenvalue $\lambda_1$ is called the {\it spectral radius} of $A$.
For the connected graph $G$, the adjacency matrix $A$ is irreducible and so there exists a unique
positive unit eigenvector corresponding to $\lambda_1$ (i.e., $\lambda_1$ has multiplicity $1$).

The \emph{ Cartesian product} $G \,\Box\, H$ of two simple undirected graphs $G$ 
and $H$ is the graph with the vertex set $V(G \,\Box\, H) = V(G) \times V(H)$ and the 
edge set
$E(G \,\Box\, H) = \{(u_i,v_k)(u_j,v_l): [(u_iu_j \in E(G))\wedge$ $(v_k = v_l)]
\vee [(v_k v_l \in E(H))\wedge (u_i = u_j)]\}$.

%
 Collatz and Sinogowitz ~\cite{CollSin57} introduced an irregularity index $\CS(G)$ and defined it as
\beq \label{eqn:001:CS}
\CS(\it{G}) &=&  \lambda_1(G) - \overline{d}(G) = \lambda_1(G) - \frac{2m}{n},
\eeq
where $\overline{d}(G) = \sum^{n}_{i=1} d_i/n = 2m/n$ denotes the average degree of the graph $G$.
%
%
%
Results of comparing ${\irr}$, $\CS$ and $\Var$ are presented in
\cite{Bell1, CvetRow88, GHM05}. 

Mukwembi~\cite{SimonMu,Mukwembi} introduced an irregularity index
$\tc(\it{G})$ of the graph $G$, as the number of distinctive terms in
the degree sequence of $G$. Clearly, for any connected graph $G$
with maximum degree $\Delta$, the irregularity index $ \tc((\it{G})$
satisfies $ \tc(\it{G}) \leq \Delta(\it{G})$.
Other attempts to determine how
irregular graph are \cite{Alavi88, Alavi87, AlaviLiu, Bell2,
Char88, Char87, Hamzeh14, JacEnt86}.

Although there have been several contributions on degree-based and
distance-based indices chemical molecular graphs, the studies on
irregularity related indices for certain special chemical
structures are still largely limited. 
In \cite{GHM-05} the irregularity of chemical trees with
respect to the variance of vertex degrees and the 
Collatz-Sinogowitz index was investigated.
The aim of the research presented in this paper is to extend that work by computing and comparing the irregularities
of some relevant chemical graphs by the four, above metionied, irregularity measures.
Specifically, the contribution of our paper is three-fold. First,
we present the irregularities of five kinds of nanostructure:
$TUC_4 C_8(S)$, $TUC_4 C_8(R)$, Zig-Zag $TUHC6$, 
$TUC_4$, Armchair $TUVC6$  nanotubes.
Then, the irregularities of dendrimer $T_{k,d}$ and circumcoronene series of benzenoid $H_k$ are deduced. At
last, we anylize the irregularities of Mycielski's constructors
$M(C_{n})$ and $M(P_{n})$.

\section[Irregularities of some chemical graphs]{Irregularities of some chemical graphs}\label{sec:wg=wlg}
%
\smallskip
%
%
%
\subsection{$TUC_4 C_8(S)[p,q]$  and $TUC_4 C_8(R)[p,q]$ nanotubes}

A $TUC_4 C_8(S)$ nanotube can be constructed by rolling a lattice of carbon atoms as it is depicted in 
Figure~\ref{nanotubeS}. The two-dimensional lattice (Figure~\ref{nanotubeS}(b)) is made by alternating squares $C_4$ and octagons $C_8$. 
We denote the number of squares in each row by $p$ and the number of rows by $q$. 
\begin{figure}[h!tb]
\begin{center}
\includegraphics[scale=0.90]{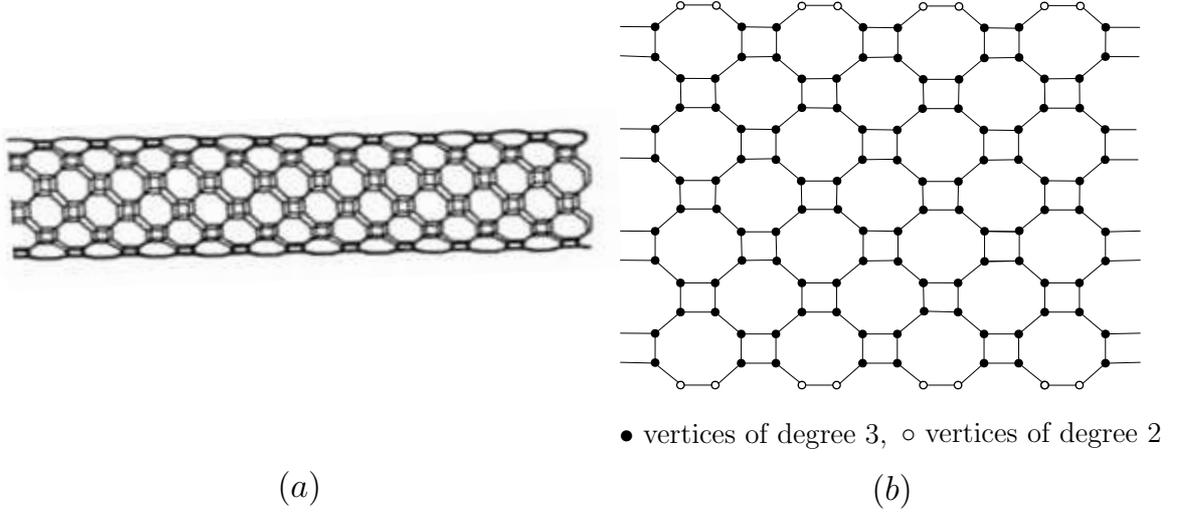}
\caption{{\small $(a)$ $3$D nanotube $TUC_4C_8(S)$, $(b)$ $2$D lattice of a $TUC_4 C_8(S)[4,4]$.}}
\label{nanotubeS}
 \end{center}
\end{figure}
\begin{te} \label{te-TUC_4C_8(S)}
Let $G=TUC_4 C_8(S)[p,q]$ be a general $TUC_4C_8(S)$ nanotube.
Then, 
$$
{ \Var}(G) =    \frac{q-1}{q^2}, \quad
{\CS}(G)= \ds  {\lambda}_{1}(G) -  3 -\frac{1}{q},\quad  {\irr}(G) =  4p, \quad  {\irr}_t(G) =   8\,p^2 ( q - 1 ).  \nonumber   
%
%
%
$$
\end{te}
\begin{proof}
It holds that  $|V(G)| = 4 p q$ and 
$|E(G)| = 2 p( 3q-1)$. Let
\beq 
V_1(G) &=& \{v \in V(G): d_G(v) = 2\},  \nonumber \\
V_2(G) &=& \{u \in V(G): d_G(u) = 3\}, \nonumber \\
E_1(G) &=& \{ e = uv \in E(G): d_G(u) \neq d_G(v)\},  \nonumber \\
E_2(G) &=& \{ e = uv \in E(G): d_G(u) = d_G(v) = 2 \},  \nonumber \\
E_3(G) &=& \{ e = uv \in E(G): d_G(u) = d_G(v) = 3\}. \nonumber 
\eeq 
Then,
\beq 
&& |V_1(G)| = 4p,\nonumber \\
&&  |V_2(G)| = 4p(q-1), \nonumber \\
&&  |E_1(G)| = 4p \mbox{ with } \imb(e) = 1, \nonumber \\
&&  |E_2(G)| = 2p \mbox{ with } \imb(e) = 0, \nonumber \\
&&  |E_3(G)| = 2p(3q-4)\mbox{ with } \imb(e) = 0. \nonumber 
\eeq 

Hence, the  variance $\Var(G)$, the
Collatz-Sinogowitz index  $\CS(G)$, the irregularity $\irr(G)$, and the total irregularity $\irr_t(G)$
of the nanotubes $TUC_4 C_8(S)[p,q]$ are  
\beq
\Var(G) &=& \,\frac{1}{n}\, \sum\limits_{v \in V(G)} d^2_G(v) - \,\frac{1}{n^2}\, (\sum\limits_{v \in V(G)} d_G(v))^2 \nonumber \\
            &=&  \frac{1}{n}(\sum\limits_{v \in V_1(G)} d^2_G(v) + \sum\limits_{v \in V_2(G)} d^2_G(v)) 
                     -\frac{1}{n^2} (\,\sum\limits_{v \in V_1(G)} d_G(v) + \,\sum\limits_{v \in V_2(G)} d_G(v)\,)^2 \nonumber \\
            &=& \frac{1}{4pq} (16p+ 36\,p (q-1)\,) - \frac{1}{16 p^2 q^2} (8 p + 12 p \,(q-1)\,)^2 
	      =    \frac{q-1}{q^2}.   \nonumber  \\
\CS(G)  &=&  \lambda_1(G) - \overline{d}(G)  =  \lambda_1(G) - \frac{2m}{n}=   \lambda_1(G) - \frac{2(2p(3q-1))}{4pq}=  \lambda_1(G) - 3 -\frac{1}{q},   \nonumber \\
\irr(G) &=& \sum\limits_{uv \in E(G)}\left| d_G(u)-d_G(v)\right|
             =  \left( \sum\limits_{uv \in E_1(G)} + \sum\limits_{uv \in E_2(G)} + \sum\limits_{uv \in E_3(G)}\right) 
                \left| d_G(u)-d_G(v)\right|  =   4p,   \nonumber\\
\irr_t(G) &=& \,\frac{1}{2}\, \sum\limits_{u,v \in V(G)}\left| d_G(u)-d_G(v)\right|
                =   \,\frac{1}{2}\, 4p\, (q-1) (4p) =  8\,p^2 ( q - 1 ).   \nonumber
\eeq
\end{proof}

%

\bigskip
\noindent
A $TUC_4 C_8(R)$ nanotube is depicted in Figure~\ref{nanotubeR} and its two-dimensional lattice is illustrated in Figure~\ref{nanotubeR}(b). 
\begin{figure}[h!tb]
\begin{center}
\includegraphics[scale=0.95]{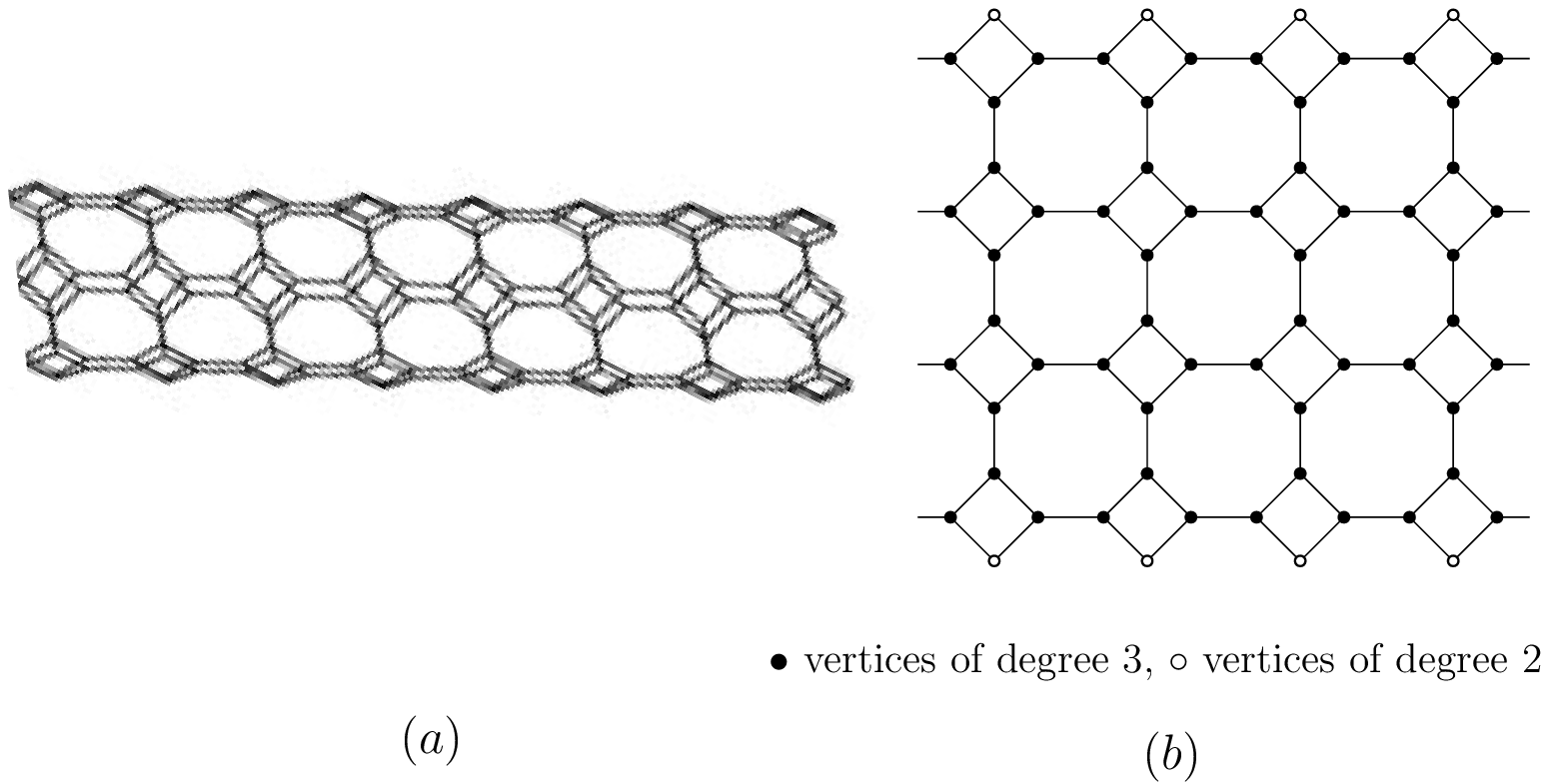}
\caption{{\small $(a)$ $3$D nanotube $TUC_4C_8(R)$, $(b)$ $2$D lattice of a $TUC_4 C_8(R)[4,4]$.}}
\label{nanotubeR}
 \end{center}
\end{figure}
\begin{te} \label{te-TUC_4C_8(R)}
Let $G=TUC_4 C_8(R)[p,q]$ be a general $TUC_4C_8(R)$ nanotube. 
Then, 
$$
{ \Var}(G) =     \frac{2q - 1 }{4 q^2}, \quad
{\CS}(G)= \lambda_1(G) -  3 - \frac{1}{2q},\quad  \irr(G) =  4p, \quad  \irr_t(G) =   2\,p^2 ( 2q - 1 ).  \nonumber   
%
%
%
$$
\end{te}
\begin{proof}
We have that that $|V(G)| = 4 p q$ and $|E(G)| = p (6q-1)$.For
\beq 
V_1(G) &=& \{v \in V(G): d_G(v) = 2\}, \nonumber \\
V_2(G) &=& \{u \in V(G): d_G(u) = 3\}, \nonumber \\
E_1(G) &=& \{ e = uv \in E(G): d_G(u) \neq d_G(v)\}, \nonumber \\
E_2(G) &=& \{ e = uv \in E(G): d_G(u) = d_G(v) = 3\}, \nonumber 
\eeq 
we have that
\beq 
&& |V_1(G)| = 2p,\nonumber \\
&& |V_2(G)| = 2p(2q-1), \nonumber \\
&& |E_1(G)| = 4p \mbox{ with } \imb(e) = 1, \nonumber \\
&& |E_2(G)| = p(6q-5)\mbox{ with } \imb(e) = 0. \nonumber 
\eeq 

The four considered irregularity measures of the nanotubes $TUC_4 C_8(R)[p,q]$ are %
\beq
\Var(G) &=& \,\frac{1}{n}\, \sum\limits_{v \in V(G)} d^2_G(v) - \,\frac{1}{n^2}\, (\sum\limits_{v \in V(G)} d_G(v))^2 \nonumber \\
         &=&  \frac{1}{n}(\sum\limits_{v \in V_1(G)} d^2_G(v) + \sum\limits_{v \in V_2(G)} d^2_G(v)) 
         	 -\frac{1}{n^2} (\,\sum\limits_{v \in V_1(G)} d_G(v) + \,\sum\limits_{v \in V_2(G)} d_G(v)\,)^2 \nonumber \\ &=& \frac{1}{4pq} (2^2 (2p)+ 3^2 (2p (2q-1))) - \frac{1}{16 p^2 q^2} (4 p + 3(2p (2q-1)))^2 
		       =    \frac{2q - 1 }{4 q^2}. \nonumber\\
\CS(G)  &=&  \lambda_1(G) - \overline{d}(G)  =  \lambda_1(G) - \frac{2m}{n}=   \lambda_1(G) - \frac{2(p(6q-1))}{4pq}=  \lambda_1(G) -  3 - \frac{1}{2q},  \nonumber \\
\irr(G) &=& \sum\limits_{uv \in E(G)}\left| d_G(u)-d_G(v)\right|
             =  \left( \sum\limits_{uv \in E_1(G)} + \sum\limits_{uv \in E_2(G)}  \right) 
                \left| d_G(u)-d_G(v)\right|  =   4p,  \nonumber \\
\irr_t(G) &=& \,\frac{1}{2}\, \sum\limits_{u,v \in V(G)}\left| d_G(u)-d_G(v)\right|
                =   \,\frac{1}{2}\, 2p\, (2q-1) (2p) =  2\,p^2 ( 2q - 1 ).   \nonumber
\eeq
\end{proof}

The computation of the adjacency matrices of $TUC_4 C_8(S)$ and $TUC_4 C_8(R)$ (and the rest of the molecular structures considered in this work) as well as 
the computation of their corresponding 
largest eigenvalues were done in Matlab. The source code for computing the adjacencies matrices is given in the appendix.
A comparison between the variance and Collatz-Sinogowitz of $TUC_4 C_8(S)$ and $TUC_4 C_8(R)$ for different values of $q$ is given in Figure~\ref{nanotubeRS}. The variance of the nanotube $TUC_4 C_8$ depends only on the number of rows $q$ (as shown in Theorems~\ref{te-TUC_4C_8(S)} and~\ref{te-TUC_4C_8(R)}). 
The computations
show that the Collatz-Sinogowitz index of $TUC_4 C_8$ depends only on the number of rows $q$, too. 
$\lambda_1(TUC_4 C_8)$.
\begin{figure}[h!]
\begin{tabular}{cc}
\includegraphics[scale=0.65]{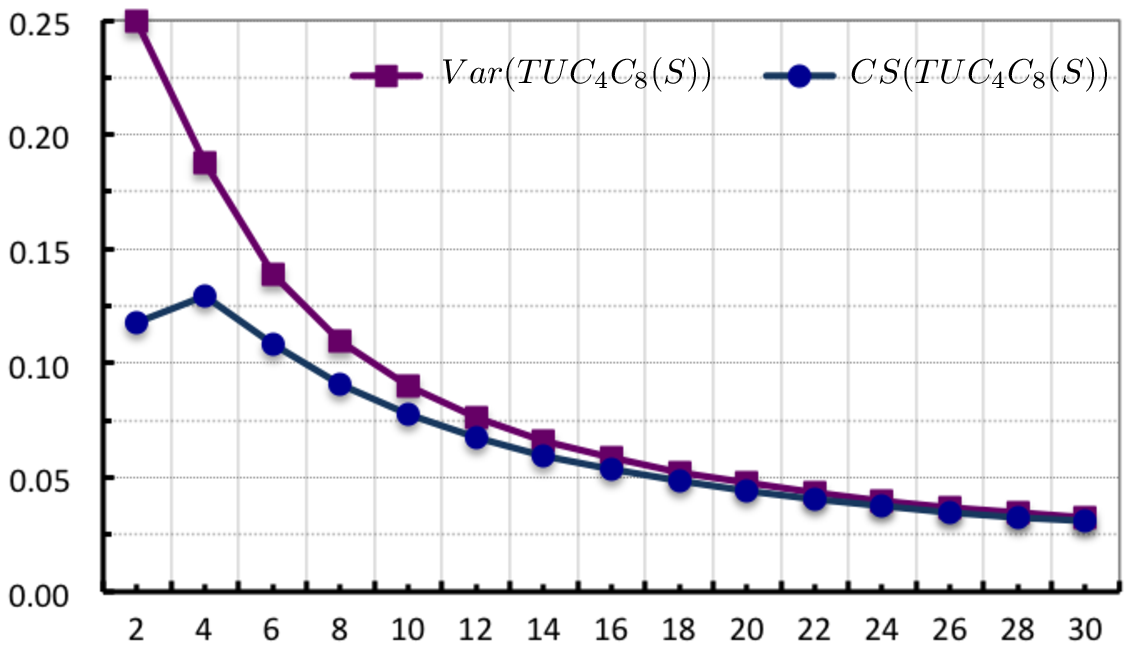} & \includegraphics[scale=0.65]{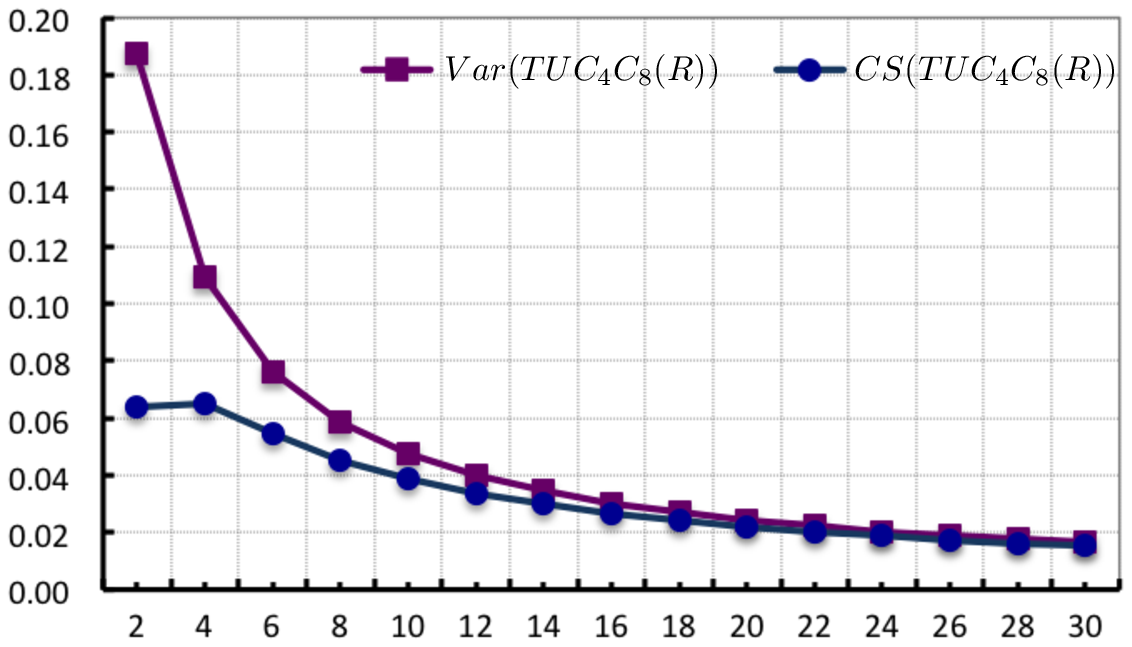} \\
(a) & (b) \\
\end{tabular}
\caption{Comparison between $\Var$ and $\CS$ of $(a)$ $TUC_4 C_8(S)$ and $(b)$ $TUC_4 C_8(R)$.}
\label{nanotubeRS}
\end{figure}

\subsection{$TUC_4(m,n)$ nanotube}
%
$TUC_4(p,q)$ is a nanotube that can be obtained as  Cartesian product of the  $p-$path $P_p$ graph and the $q-$cycle graph $C_q$ (Figure~\ref{TUC4Fig}). 
We denote the number of vertices in a row by $p$ and the number of vertices in a column by $q$.

\begin{figure}[ht!p]
\begin{center}
\includegraphics[scale=0.85]{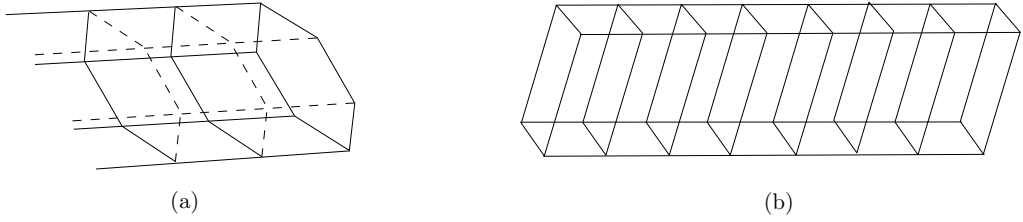}
\caption{{\small  $(a)$ Nanotubes $TUC_4[p,6]$, $(b)$ Nanotubes $TUC_4[8,4]$}}
\label{TUC4Fig}
\end{center}
\end{figure}

\begin{te} \label{te-TUC_4}
Let $G=TUC_4(p,q)$. 
Then, 
$$
{ \Var}(G) =     \frac{2(p-2)}{p^2}, \quad
{\CS}(G)=  \lambda_1(G) - 4 + \frac{2}{p},\quad  \irr(G) =  2q, \quad  \irr_t(G) =   q^2 ( p - 2 ).  \nonumber   
%
%
%
$$
\end{te}

\begin{proof}
 It holds that  $|V(G)| = p q$ and $|E(G)| = q( 2 p -1)$. Let
\beq 
V_1(G) &=& \{v \in V(G): d_G(v) = 3\}, \nonumber \\
V_2(G) &=& \{u \in V(G): d_G(u) = 4\}, \nonumber \\
E_1(G) &=& \{ e = uv \in E(G): d_G(u) \neq d_G(v)\}, \nonumber \\
E_2(G) &=& \{ e = uv \in E(G): d_G(u) = d_G(v) = 3 \}, \nonumber \\
E_3(G) &=& \{ e = uv \in E(G): d_G(u) = d_G(v) = 4\}. \nonumber 
\eeq 
Then,
\beq 
&&  |V_1(G)| = 2q,\nonumber \\
&& |V_2(G)| = (p-2)q, \nonumber \\
&& |E_1(G)| = 2q\mbox{ with } \imb(e) = 1, \nonumber \\
&& |E_2(G)| = 2q\mbox{ with } \imb(e) = 0, \nonumber \\
&& |E_3(G)| = q(2p-5)\mbox{ with } \imb(e) = 0. \nonumber 
\eeq 
Consequently, the four irregularity measures of the nanotubes $TUC_4(p,q)$ are  
\beq
\Var(G) &=& \,\frac{1}{n}\, \sum\limits_{v \in V(G)} d^2_G(v) - \,\frac{1}{n^2}\, (\sum\limits_{v \in V(G)} d_G(v))^2 \nonumber \\
            &=&  \frac{1}{n}(\sum\limits_{v \in V_1(G)} d^2_G(v) + \sum\limits_{v \in V_2(G)} d^2_G(v)) 
                     -\frac{1}{n^2} (\,\sum\limits_{v \in V_1(G)} d_G(v) + \,\sum\limits_{v \in V_2(G)} d_G(v)\,)^2 \nonumber \\
            &=& \frac{1}{pq} (3^2 (2q)+ 4^2 (p-2)q) - \frac{1}{p^2 q^2} (3 (2q) + 4 (p-2)q)^2 
	      =    \frac{2(p-2)}{p^2}.  \nonumber \\
\CS(G)  &=&  \lambda_1(G) - \overline{d}(G)  =  \lambda_1(G) - \frac{2m}{n}=   \lambda_1(G) - \frac{2(q(2p-1))}{pq}=  \lambda_1(G) - 4 + \frac{2}{p}, \nonumber \\
\irr(G) &=& \sum\limits_{uv \in E(G)}\left| d_G(u)-d_G(v)\right|
             =  \left( \sum\limits_{uv \in E_1(G)} + \sum\limits_{uv \in E_2(G)} + \sum\limits_{uv \in E_3(G)}\right)  \left| d_G(u)-d_G(v)\right|  =   2q,   \nonumber \\
\irr_t(G) &=& \,\frac{1}{2}\, \sum\limits_{u,v \in V(G)}\left| d_G(u)-d_G(v)\right|
                =   \,\frac{1}{2}\, 2q\, (p-2) q =  q^2 ( p - 2 ).    \nonumber
\eeq
\end{proof}

A comparison between the variance and Collatz-Sinogowitz of $TUC_4$  for different values of $q$ is given in Figure~\ref{comparVarColl2}. The variance of the nanotube $TUC_4 $ depends only on the number of rows $p$ (as shown in Theorem~\ref{te-TUC_4}). 
\begin{figure}[h!]
\begin{tabular}{cc}
\includegraphics[scale=0.65]{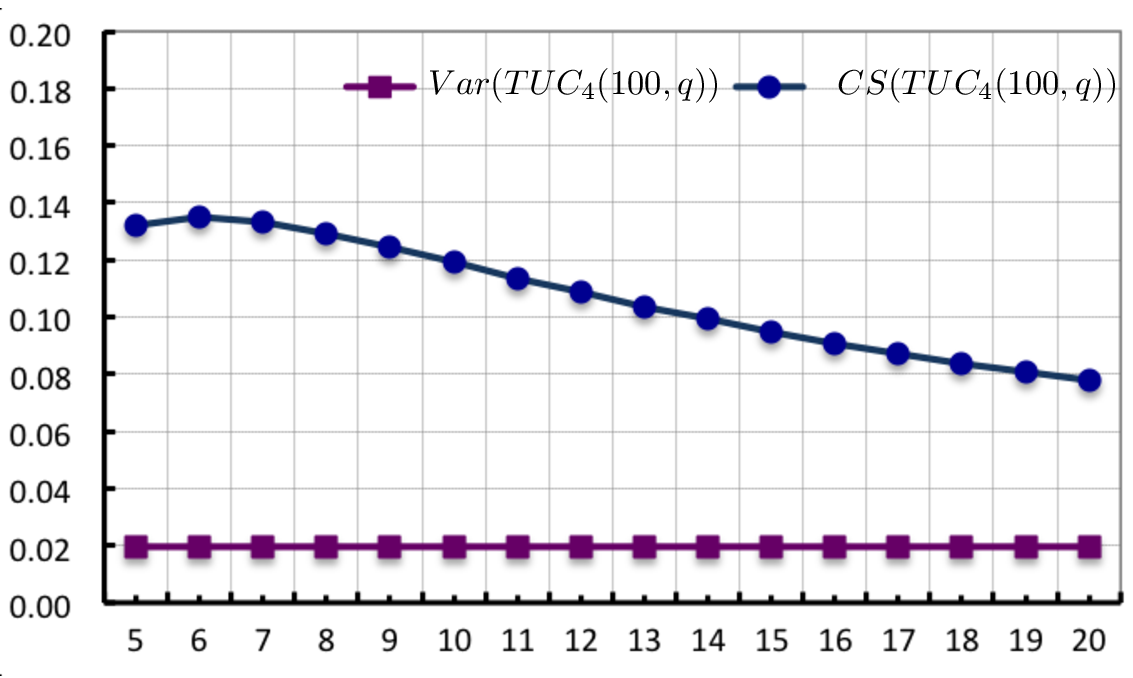} & \includegraphics[scale=0.65]{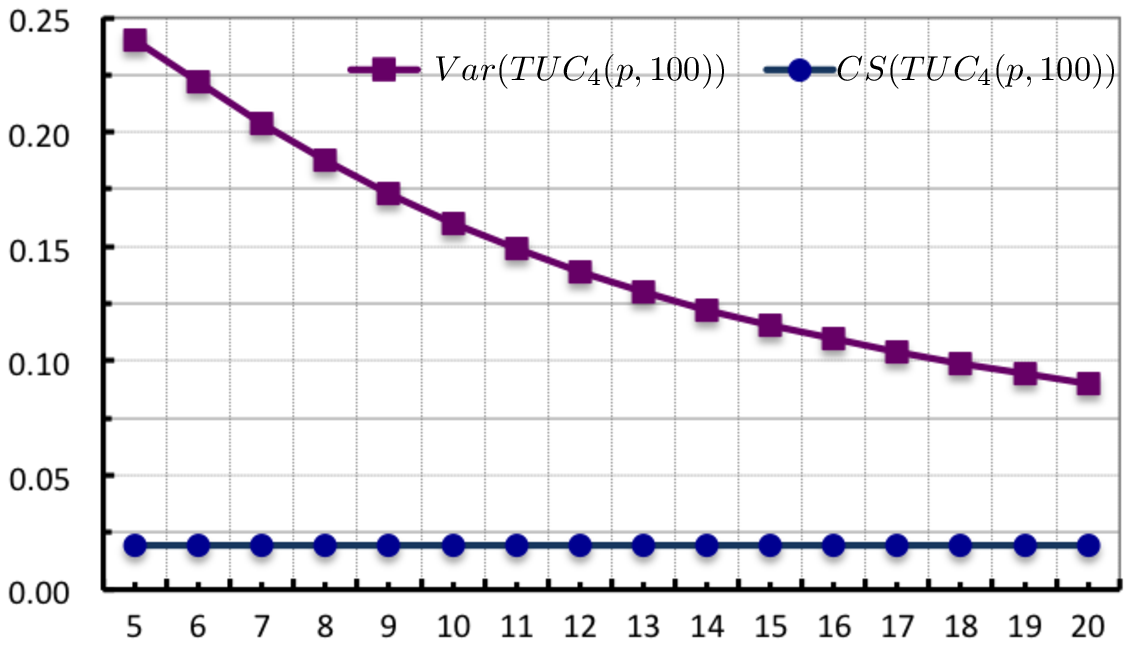} \\
(a) & (b) \\
\end{tabular}
\caption{Comparison between $\Var$ and $\CS$ of $(a)$ $TUC_4(100,q)$ and $(b)$ $TUC_4 (p,100)$.}
\label{comparVarColl2}
\end{figure}
Observe that $\Var( TUC_4 (p,q))= 2(p-2)/p^2$ and it is independent of $q$. Therefore, $\Var( TUC_4 (100,q))$ has a constant value of  $0.0196$. 
The calculations show that $\CS(TUC_4(p,100))$ is independent of $p$, respectively. However, the theoretical proof of this statement is missing.

\subsection{Zig-Zag $TUHC_{6}$ nanotube}
Let $G = TUHC_{6}[p,q]$ be a Zig-Zag polyhex nanotube, where $p$ is the number of hexagons in each row 
and $q$ is the number of Zig-Zag lines in the molecular graph of $G$, as it is
depicted in Figure~(\ref{ZigZag_nanotubes}). 

\begin{figure}[htbp!]
\begin{center}
\begin{minipage}[h]{1.0\textwidth}
\begin{center} 
\includegraphics[scale=0.85]{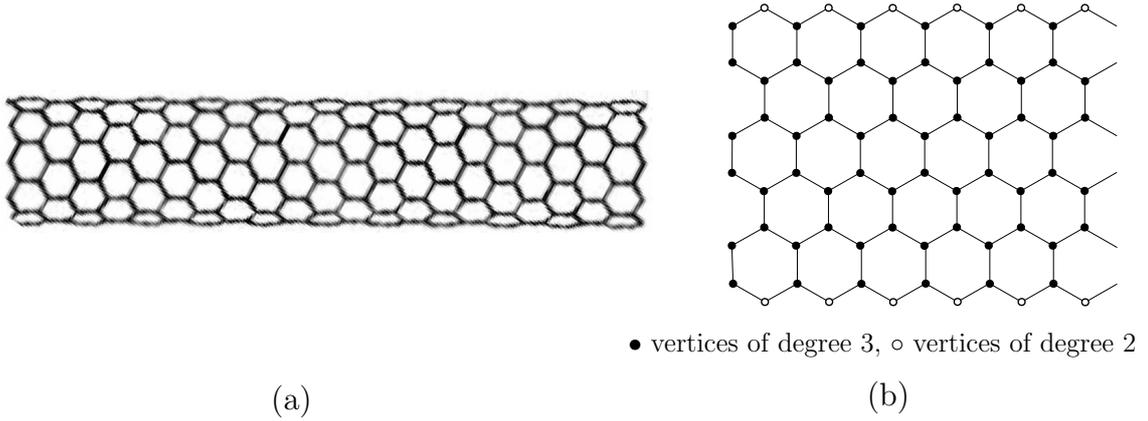}
\caption{{\small  $(a)$ $3$D nanotube $TUHC6[p,q]$, $(b)$ $2$D lattice of a $TUHC6[6,6]$.}}
\label{ZigZag_nanotubes}
\end{center} 
\end{minipage}
\end{center}
\end{figure}

\begin{te} \label{te-TUHC6}
Let $G = TUHC_{6}[p,q]$ be a be a Zig-Zag polyhex nanotube. 
Then, 
$$
{ \Var}(G) =      \frac{q-1}{q^2}, \quad
{\CS}(G)= \lambda_1(G) - 3 + \frac{1}{q},\quad  \irr(G) =  4p, \quad  \irr_t(G) =   4 p^2 (q-1).  \nonumber   
%
%
%
$$
\end{te}

\begin{proof}
We have that $|V(G)| = 2 p q$ and $|E(G)| = p (3 q - 1)$. For
\beq 
V_1(G) &=& \{v \in V(G): d_G(v) = 2\},\nonumber \\
V_2(G) &=& \{u \in V(G): d_G(u) = 3\}, \nonumber \\
E_1(G) &=& \{ e = uv \in E(G): d_G(u) \neq d_G(v)\}, \nonumber \\
E_3(G) &=& \{ e = uv \in E(G): d_G(u) = d_G(v) = 3\}, \nonumber 
\eeq 
it follows that
\beq 
&& |V_1(G)| = 2p,\nonumber \\
&& |V_2(G)| = 2p(q-1), \nonumber \\
&& |E_1(G)| = 4p \mbox{ with } \imb(e) = 1, \nonumber \\
&& |E_3(G)| = p(3q-5)\mbox{ with } \imb(e) = 0. \nonumber 
\eeq 
Thus, the variance $\Var(G)$,
the Collatz-Sinogowitz index, the  irregularity, and the total irregularity
of the nanotubes $TUHC_{6}[p,q]$ are  

\beq
\Var(G) &=& \frac{1}{n}\, \sum\limits_{v \in V(G)} d^2_{G}(v) - \,\frac{1}{n^2}\,(\sum\limits_{v \in V(G)} d_{G}(v))^2 \nonumber \\
             &=& \frac{1}{n}(\sum\limits_{v \in V_1(G)} d^2_{G}(v) + \sum\limits_{u \in V_2(G)} d^2_{G}(u))
                   - \frac{1}{n^2}(\sum\limits_{v \in V_1(G)} d_{G}(v) + \sum\limits_{v \in V_2(G)} d_{G}(v)\,)^2 \nonumber \\
             &=& \frac{1}{2pq}\,(2^2 . 2p  + 3^2 . 2p\,(q-1)) - \frac{1}{4 p^2\,q^2}\,(2. 2p + 3.2p(q-1))^2
                = \frac{q-1}{q^2}. \nonumber \\
\CS(G) &=& \lambda_1(G) - \overline{d}(G)  =  \lambda_1(G) - \frac{2m}{n}=   \lambda_1(G) - \frac{2(p(3q-1))}{2pq}=  \lambda_1(G) - 3 + \frac{1}{q},\\
 \nonumber \\
\irr(G)&=& \sum\limits_{e \in E(G)}\left| d_{G}(u)-d_{G}(v)\right|
             =  \sum\limits_{e \in E'(G)}\left| d_{G}(u)-d_{G}(v)\right|=    4p. \nonumber \\
\irr_t(G)&=& \,\frac{1}{2}\,\sum\limits_{u \in V(G)} \sum\limits_{v \in V(G)} \left| d_{G}(u)-d_{G}(v)\right|
               =  2p(2p(q-1)) = 4 p^2 (q-1). \nonumber
\eeq
\end{proof}

A comparison between the variance and Collatz-Sinogowitz of $TUHC_{6}[p,q]$  for different values of $q$ is given in Figure~\ref{comparVarCollTUHC6}. The variance of the nanotube $TUHC_6$ depends only on the parameter $q$ (as shown in Theorem~\ref{te-TUHC6}). 
\begin{figure}[h!]
\begin{center} 
\begin{tabular}{c}
\includegraphics[scale=0.85]{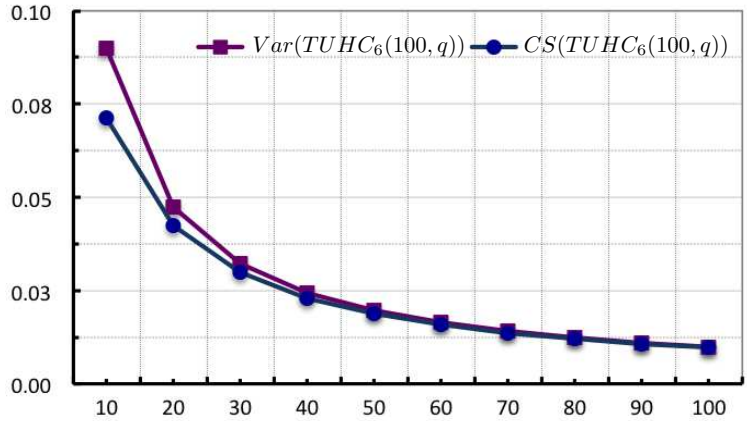}  \\
\end{tabular}
\caption{Comparison between $\Var$ and $\CS$ of $TUHC_6(100,q)$ .}
\label{comparVarCollTUHC6}
\end{center} 
\end{figure}
%
\subsection{$TUVC_{6}$ nanotube}
Armchair $TUVC_{6}[p,q]$ nanotube can be constructed by rolling a lattice of carbon atoms
comprised of $q$ columns and $p$ hexagons in each row (Figure~\ref{Armchair}).

\begin{figure}[htbp!]
\begin{center}
\begin{minipage}[h]{1.0\textwidth}
\begin{center} 
\includegraphics[scale=0.85]{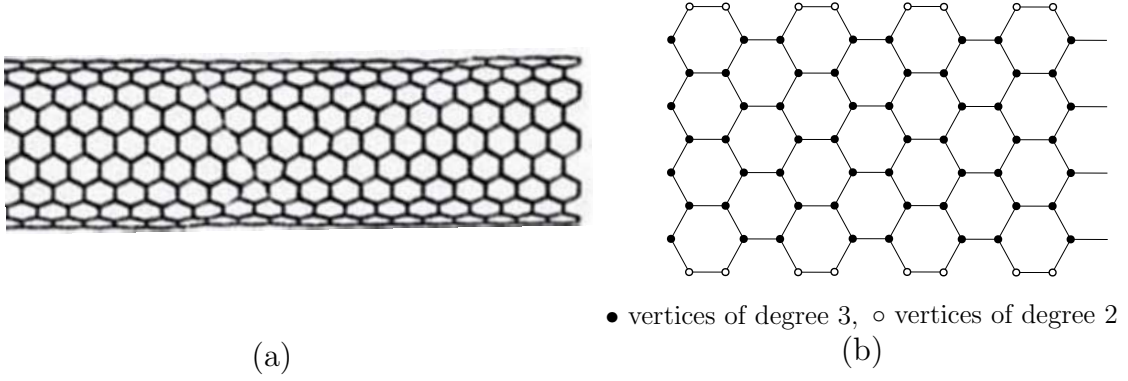}
\caption{{\small  $(a)$ Armchair $3$D nanotube $TUVC6[p,q]$, $(b)$ $2$D lattice of a $TUVC6[4,9]$.}}
\label{Armchair}
\end{center} 
\end{minipage}
\end{center}
\end{figure}
%

\begin{te} \label{te-TUVC6}
Let $G = TUVC_{6}[p,q]$ be an arbitrary armchair polyhex nanotube.
Then, 
$$
{ \Var}(G) =     \frac{2(q-2)}{q^2}, \quad
{\CS}(G)= \lambda_1(G) -  3 + \frac{2}{q},\quad  \irr(G) =  4p, \quad  \irr_t(G) =   4\,p^2 ( q - 2 ).  \nonumber   
%
%
%
$$
\end{te}

\begin{proof}
 It holds that  $|V(G)| = 2 p q$ and $|E(G)| =  p( 3q-2)$. Let,
\beq 
V_1(G) &=& \{v \in V(G): d_G(v) = 2\},\nonumber \\
V_2(G) &=& \{u \in V(G): d_G(u) = 3\}, \nonumber \\
E_1(G) &=& \{ e = uv \in E(G): d_G(u) \neq d_G(v)\}, \nonumber \\
E_2(G) &=& \{ e = uv \in E(G): d_G(u) = d_G(v) = 2 \}, \nonumber \\
E_3(G) &=& \{ e = uv \in E(G): d_G(u) = d_G(v) = 3\}. \nonumber 
\eeq 
Then,
\beq 
&& |V_1(G)| = 4p,\nonumber \\
&& |V_2(G)| = 2p(q-2), \nonumber \\
&& |E_1(G)| = 4p \mbox{ with } \imb(e) = 1, \nonumber \\
&& |E_2(G)| = 2p \mbox{ with } \imb(e) = 0, \nonumber \\
&& |E_3(G)| = p(3q-8)\mbox{ with } \imb(e) = 0. \nonumber 
\eeq 

Consequently, the all four irregularity measures: variance $\Var(G)$,
Collatz-Sinogowitz index $\CS(G)$, irregularity $\irr(G)$, and the total irregularity $\irr_t(G)$
of the nanotubes $TUVC_{6}[p,q]$ are  
\beq
\Var(G) &=& \,\frac{1}{n}\, \sum\limits_{v \in V(G)} d^2_G(v) - \,\frac{1}{n^2}\, (\sum\limits_{v \in V(G)} d_G(v))^2 \nonumber \\
            &=&  \frac{1}{n}(\sum\limits_{v \in V_1(G)} d^2_G(v) + \sum\limits_{v \in V_2(G)} d^2_G(v)) 
                     -\frac{1}{n^2} (\,\sum\limits_{v \in V_1(G)} d_G(v) + \,\sum\limits_{v \in V_2(G)} d_G(v)\,)^2 \nonumber \\
            &=& \frac{1}{2pq} (16p+ 18\,p (q-2)\,) - \frac{1}{4 p^2 q^2} (8 p + 6 p \,(q-2)\,)^2 
	      =    \frac{2(q-2)}{q^2}.  \label{var0TUVC} \nonumber  \\
\CS(G)  &=&  \lambda_1(G) - \overline{d}(G)  =  \lambda_1(G) - \frac{2m}{n}=   \lambda_1(G) - \frac{2(2p(3q-1))}{4pq}=  \lambda_1(G) -  3 + \frac{2}{q},  \nonumber\\
\irr(G) &=& \sum\limits_{uv \in E(G)}\left| d_G(u)-d_G(v)\right|
             =  \left( \sum\limits_{uv \in E_1(G)} + \sum\limits_{uv \in E_2(G)} + \sum\limits_{uv \in E_3(G)}\right) 
                \left| d_G(u)-d_G(v)\right|  =   4p,   \nonumber\\
\irr_t(G) &=& \,\frac{1}{2}\, \sum\limits_{u,v \in V(G)}\left| d_G(u)-d_G(v)\right|
                =   \,\frac{1}{2}\, 4p\, (q-2) (2p) =  4\,p^2 ( q - 2 ).   \nonumber
\eeq
\end{proof}

A comparison between the variance and Collatz-Sinogowitz of $TUVC_{6}[p,q]$  for different values of $q$ is given in Figure~\ref{comparVarCollTUVC6}. The variance of the nanotube $TUVC_4 $ depends only on the parameter $q$ (as shown in Theorem~\ref{te-TUVC6}). 
\begin{figure}[h!]
\begin{center} 
\begin{tabular}{c}
\includegraphics[scale=0.85]{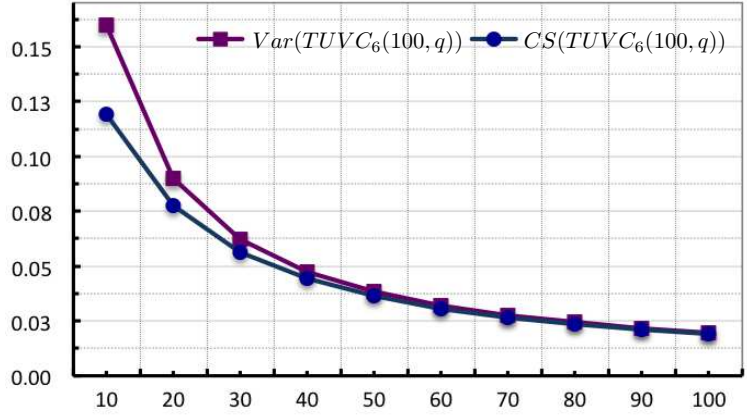}  \\
\end{tabular}
\caption{Comparison between $\Var$ and $\CS$ of $TUVC_6(100,q)$ .}
\label{comparVarCollTUVC6}
\end{center} 
\end{figure}
%
\subsection{$T_{k,d}$  dendrimer}
%
%

A tree $T$ is a {\it complete $k$-regular} if every vertex has
degree $1$ or $k$. A tree where all leaves are on the same
distance to the root is called a {\it balanced tree}. By
$T_{k,d}$, we denote a balanced $k$-regular tree whose leaves are
at distance $d$ to the root of the tree. In chemical graph theory
$T_{k,d}$ trees are also known as $T_{k,d}$ {\it dendrimers}. The
$T_{k,d}$ dendrimers, for several different parameters of $k$ and $d$, are illustrated in Figure~\ref{Tree34}.

\begin{figure}[!h]
\begin{center}
\includegraphics[scale=0.85]{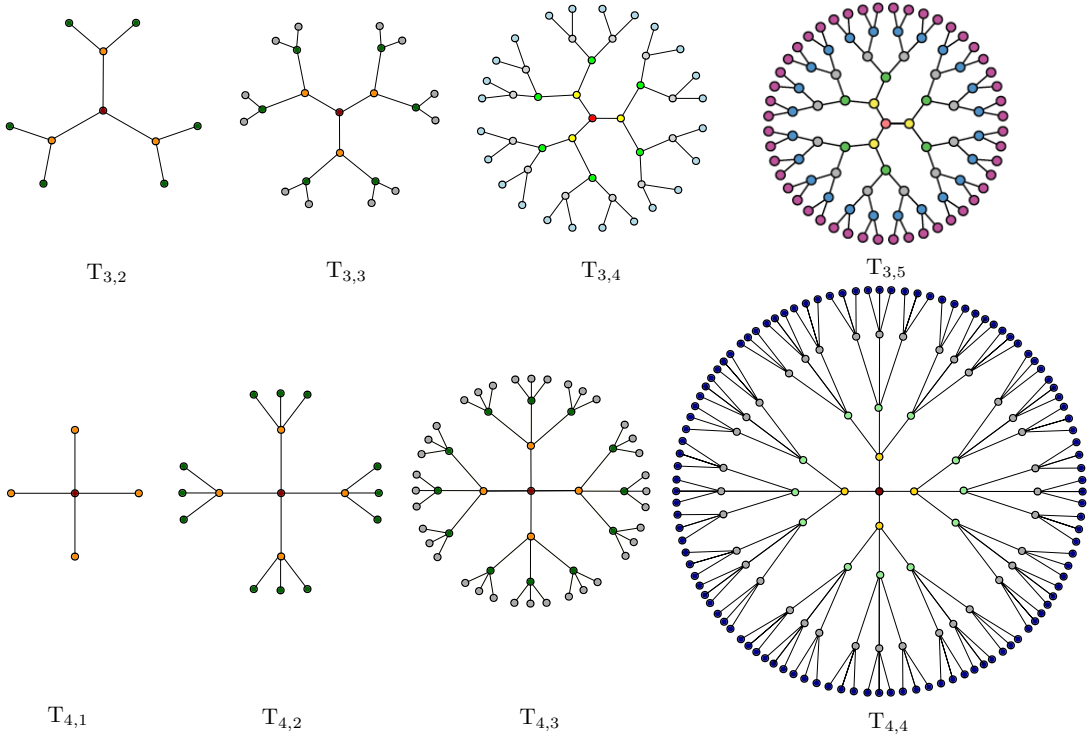}
\caption{Molecular graphs of dendrimers $T_{3,0} ,\cdots , T_{3,3}$ and $T_{4,0} , \cdots , T_{4,3}$.}
\label{Tree34}
\end{center}
\end{figure}

\begin{te} \label{te-dendrimer}
Let $T_{k,d}$ be an a dendrimer. Then, 
\beq
\begin{aligned}[t]
&{ \Var}(T_{k,d}) =    \frac{k\,(k - 2 )\,( k - 1 )^d\,(k\,((k-1)^d-2)+2)}{( k (k-1)^d - 2 )^2}, 
&& \quad {\CS}(T_{k,d})=   \lambda_1(T_{k,d}) - \frac{2k ((k-1)^d-1)}{k(k-1)^d-2} ,  \nonumber \\ 
 & \irr(T_{k,d}) =  k (k-1)^d, 
 &&\quad  \irr_t(T_{k,d}) =    \frac{\, k^2\, (k-1)^d \,(\,(k-1)^{d-1}-2)}{2(k-2)}.  \nonumber   
 \end{aligned}
\eeq
\end{te}

\begin{proof}
It can be easily computed that
$|V(T_{k,d})| = (k(k-1)^d-2)/(k-2)$ and $|E(T_{k,d})| = k ((k-1)^d-1)/(k-2)$.
Let $V_1(T_{k,d}) =$$\{ v \in V(T_{k,d}) :\; $$d_{T_{k,d}}(v) = 1\}$, and
$V_2(T_{k,d}) =$$\{ u \in V(T_{k,d}) :\; $$d_{T_{k,d}}(u) = k\}$. Then
$|V_1(T_{k,d})| = k (k-1)^{d-1}$ and $|V_2(T_{k,d})| = (k(k-1)^{d-1}-2)/(k-2)$.
Let $E_1(T_{k,d}) =$$\{ uv \in E(T_{k,d}) :\; $$d_{T_{k,d}}(u) \neq d_{T_{k,d}}(v)\}$.
$|E_1({T_{k,d}})|= k (k-1)^{d-1}$, which is the number of leaves of $T_{k,d}$ with  
${\imb(uv)} = k-1$ for all $uv \in E_1(G)$.

Thus, the variance, the Collatz-Sinogowitz index,
the irregularity  and total irregularity of $T_{k,d}$ are
\beq
\Var(T_{k,d})&=& \frac{1}{n}\, \sum\limits_{v \in V(T_{k,d})} d^2_{T_{k,d}}(v) - \,\frac{1}{n^2}\,
                              (\sum\limits_{v \in V(T_{k,d})} d_{T_{k,d}}(v))^2 \nonumber \\
             &=& \frac{1}{n}(\sum\limits_{v \in V_1(T_{k,d})} d^2_{T_{k,d}}(v)
                    +\sum\limits_{u \in V_2(T_{k,d})} d^2_{T_{k,d}}(u)) \nonumber \\
              & & -\frac{1}{n^2}(\sum\limits_{v \in V_1(T_{k,d})} d_{T_{k,d}}(v)
                    + \sum\limits_{v \in V_2(T_{k,d})} d_{T_{k,d}}(v))^2 \nonumber \\
              &=&   \,\frac{k-2}{k (k-1)^d - 2}\,\left(k\,(k-1)^{d-1}\,+\,k^2 \,\frac{k\,(k-1)^{d-1}-2}{k-2} \right)\nonumber \\
              & &  -\,\left(\,\frac{k-2}{k \,(k-1)^d - 2}\right)^2 \,
                            \left( k\,(\,k - 1)^{d-1} \,+\,k \, \frac{\,k \,(k - 1)^{d - 1}-2}{k-2}\,\right)^2  \nonumber \\
               &=&  \frac{k\,(k - 2 )\,( k - 1 )^d\,(k\,((k-1)^d-2)+2)}{( k (k-1)^d - 2 )^2}.      \nonumber \\
\CS(T_{k,d}) &=&  \lambda_1(T_{k,d}) - \frac{2k ((k-1)^d-1)}{k(k-1)^d-2} \nonumber  \\   && \nonumber \\
%
%
\irr(T_{k,d}) &=& \sum\limits_{e \in E(T_{k,d})}\left| d_{T_{k,d}}(u)-d_{T_{k,d}}(v)\right|
                     =    \sum\limits_{e \in E'(T_{k,d})}\left| d_{T_{k,d}}(u)- d_{T_{k,d}}(v)\right| = k (k-1)^d.  \nonumber \\
%
\irr_t(T_{k,d}) &=&\frac{1}{2}\,\sum\limits_{u,v \in V(T_{k,d})} \left| d_{T_{k,d}}(u)-d_{T_{k,d}}(v)\right|               
                        = \frac{\, k^2\, (k-1)^d \,(\,(k-1)^{d-1}-2)}{2(k-2)}.       \nonumber
\eeq
\end{proof}
\begin{figure}[h!]
\begin{tabular}{cc}
\includegraphics[scale=0.65]{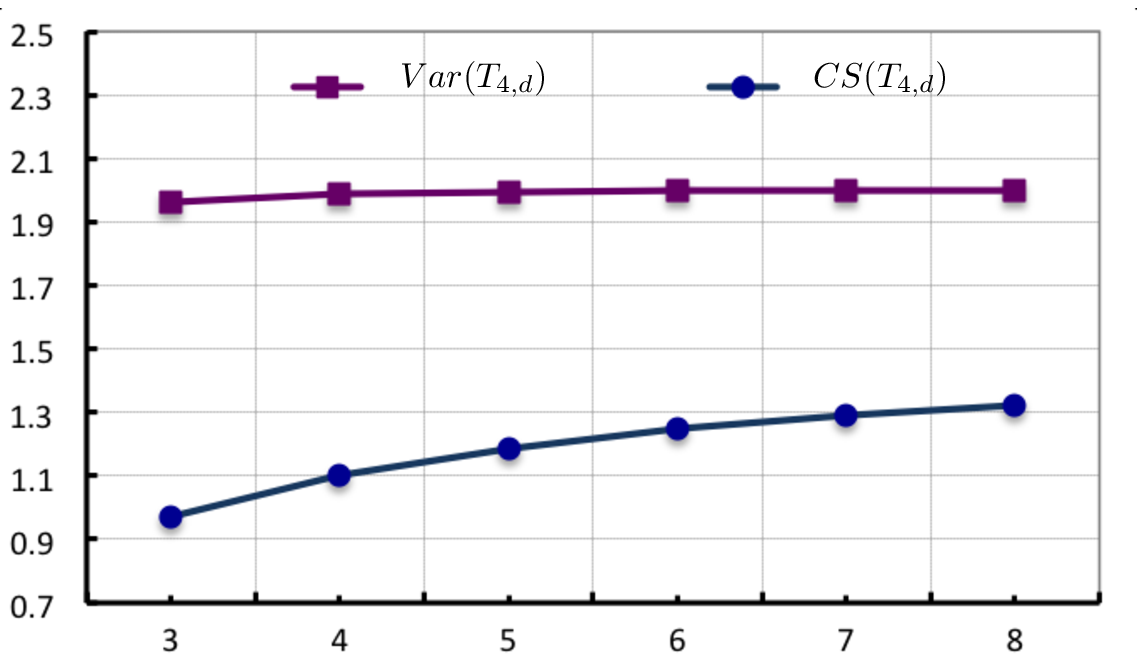} & \includegraphics[scale=0.65]{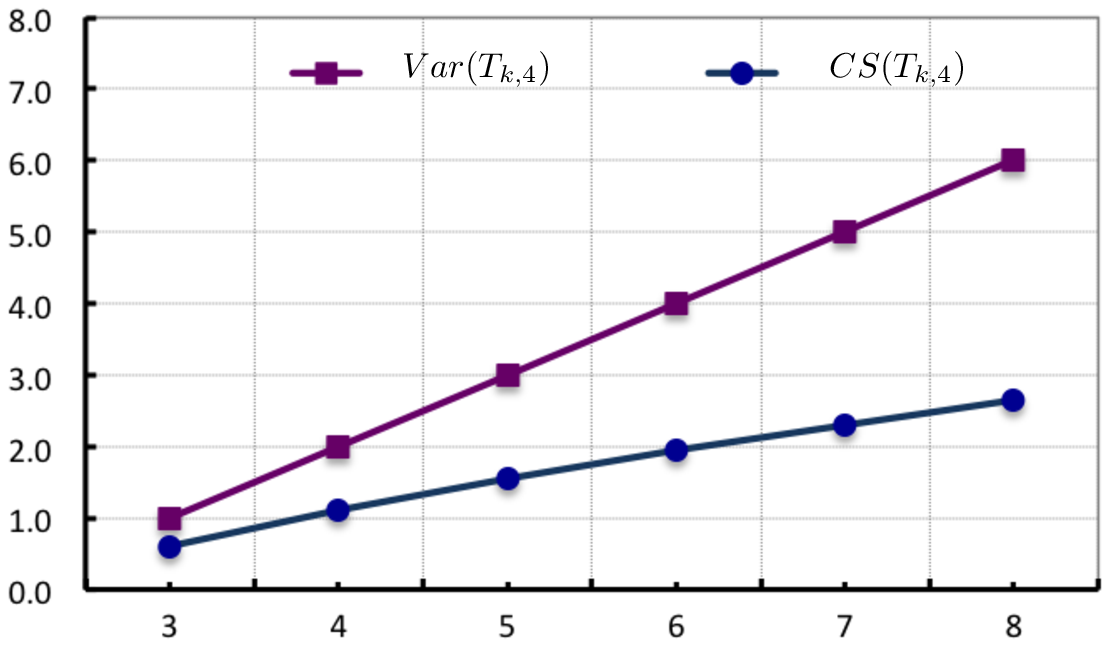} \\
(a) & (b) \\
\\
\includegraphics[scale=0.65]{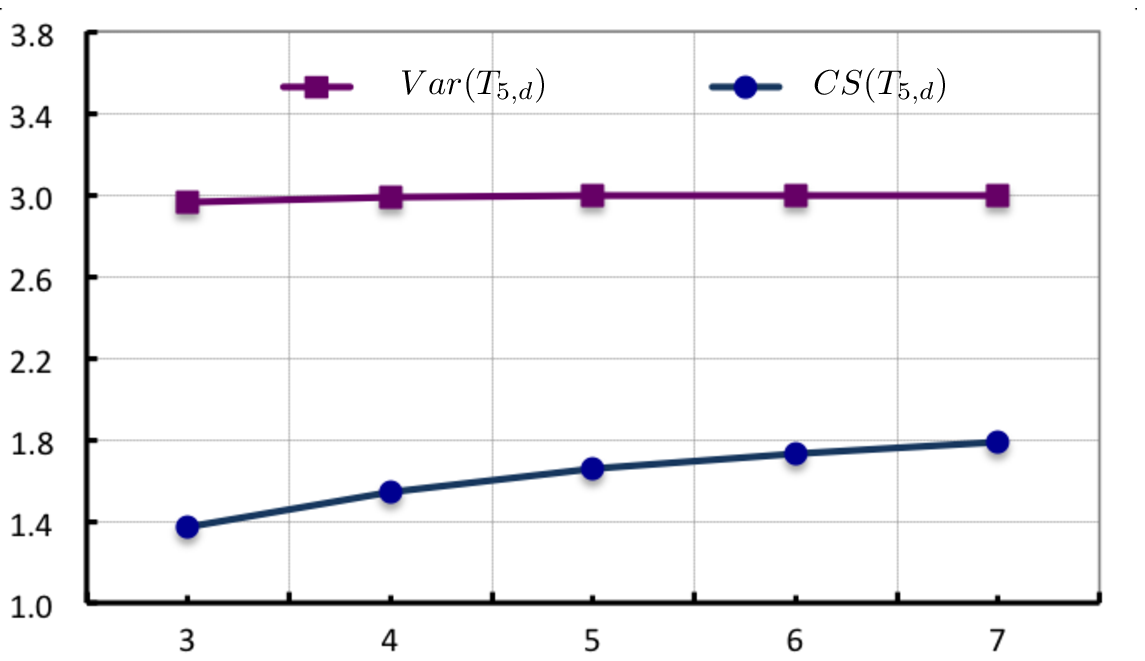} & \includegraphics[scale=0.65]{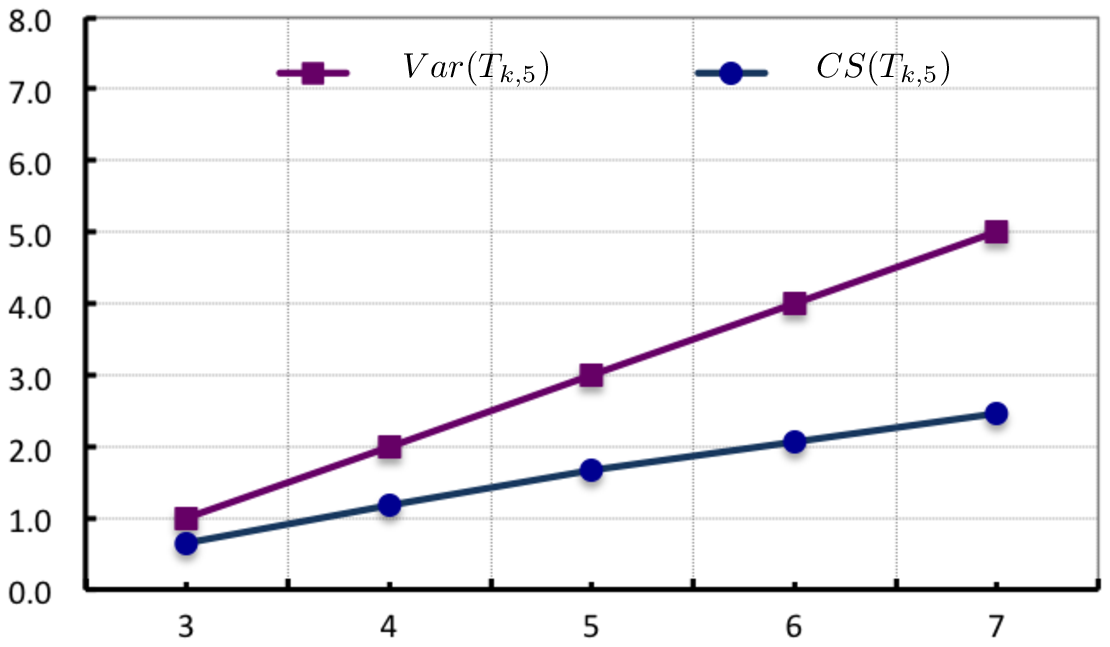} \\
(c) & (d) \\
\end{tabular}
\caption{Comparison between $\Var$ and $\CS$ of $(a)$ $T_{4,d}$, $(b)$ $T_{k,4}$, $(c)$ $T_{5,d}$ and $(d)$ $T_{k,5}$.}
\label{comparVarColl2}
\end{figure}
%
\subsection{Circumcoronene series of benzenoid $H_k$}
%
 \begin{figure}[!h]
\begin{center}
\includegraphics[scale=0.85]{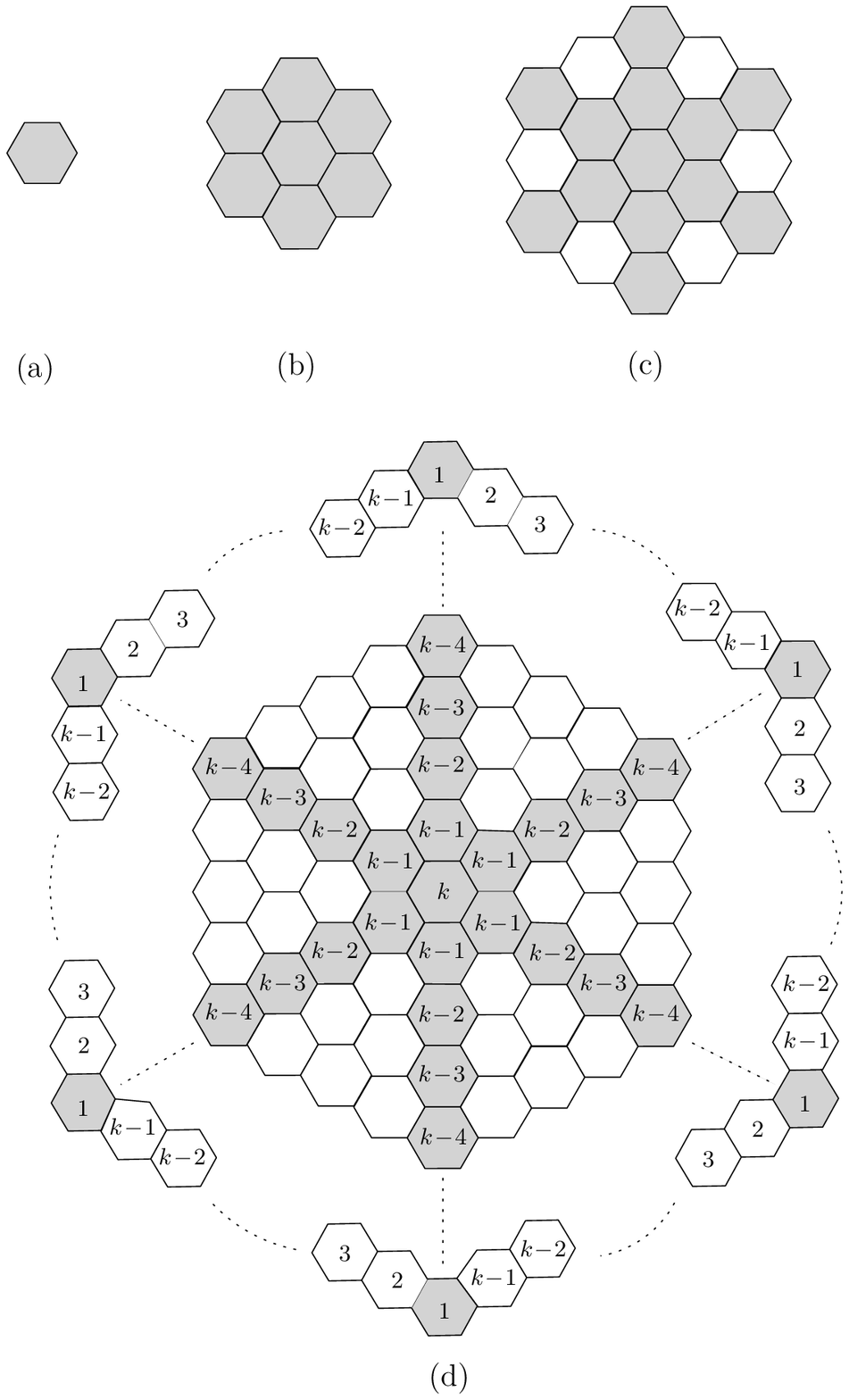}
\caption{{\small Circumcoronene series of benzenoid $H_k$. In $(a)$  Benzene $ H_1 = C_6$, $(B)$  Coronene $H_2 = C_6(C_6)$, $(c)$ Circumcoronene $H_3$ and (d) The circumcoronene series of benzenoid $H_k$.}}
\label{Circumcoronene}
\end{center}
\end{figure}

 In Figure (\ref{Circumcoronene}) the circumcoronene series of benzenoid $H_k$, 
 for $k = 1, 2, 3$ and the circumcoronene series in the general case are depicted. 
 The structures of this family of circumcoronene are presented as
 homologous series of benzenoid consisted several copy of benzene $C_6$ on circumference.
 Consider circumcoronene series of benzenoid $H_k$ for $k \geq 1$. 
 It holds that  $|V(H_k)| = 6 k^2$ and $|E(H_k)| =  3k ( 3k - 1)$. 

\begin{te} \label{te-Circumcoronene}
Let $H_k$ be an a Circumcoronene. Then, 
\beq
{ \Var}(H_k) =    \frac{k-1}{k^2}, 
 \quad {\CS}(H_k,)=   \lambda_1(T_{k,d}) - \frac{3k-1}{k},  
 \quad  \irr(H_k) =  12 (k-1), 
 \quad  \irr_t(H_k) =   36\, k^2\, (k-1). \nonumber   
\eeq
\end{te}
\begin{proof}
A direct calculations gives that
$|V(H_k)| = 6 k^2$ and $|E(H_k)| =  3k ( 3k - 1)$.
Let $V_1(H_k)$$ = \{ v \in V(H_k) :\; $$d_{H_k}(v) = 2\}$, and
$V_2(H_k) =$$\{ u \in V(H_k) :\; $$d_{H_k}(u) = 3\}$. Then
$|V_1(H_k)| = 6 k $ and $|V_2(H_k)| = 6 k (k-1)$.
Let $E_1(H_k) =$$\{ uv \in E(H_k) :\; $$d_{H_k}(u) \neq d_{H_k}(v)\}$ with  
${\imb(uv)} = 1$ for all $uv \in E_1(H_k)$,  $E_2(H_k) =$$\{ uv \in E(H_k) :\; $$d_{H_k}(u) = d_{H_k}(v)=2\}$ with  
${\imb(uv)} = 0$ for all $uv \in E_2(H_k)$ and $E_3(H_k) =$$\{ uv \in E(H_k) :\; $$d_{H_k}(u) = d_{H_k}(v)=3\}$ with  
${\imb(uv)} = 0$ for all $uv \in E_3(H_k)$. Then $|E_1({H_k})|= 12 (k-1)$, $|E_2({H_k})|= 6$ and $|E_3({H_k})|= 3 (3k-2) (k-1)$.

Thus, the variance, the Collatz-Sinogowitz index, the irregularity  and total irregularity of $H_k$ are
\beq
\Var(H_k)&=& \frac{1}{n}\, \sum\limits_{v \in V(H_k)} d^2_{H_k}(v) - \,\frac{1}{n^2}\,
                              (\sum\limits_{v \in V(H_k)} d_{H_k}(v))^2 \nonumber \\
             &=& \frac{1}{n}(\sum\limits_{v \in V_1(H_k)} d^2_{H_k}(v)
                    +  \sum\limits_{u \in V_2(H_k)} d^2_{H_k}(u)) - \frac{1}{n^2}(\sum\limits_{v \in V_1(H_k)} d_{H_k}(v)
                    + \sum\limits_{v \in V_2(H_k)} d_{H_k}(v))^2 \nonumber \\
              &=&   \,\frac{2^2 (6k) + 3^2 (6k(k-1))}{6 k^2 }\, - \,\frac{( 2 (6k) + 3 ( 6k (k-1) ) )^2}{36 k^4}    \nonumber \\
              &=&  \frac{9k-5}{k} - \frac{9k^2-6k+1}{k^2}      \nonumber \\
              &=&  \frac{k-1}{k^2}.      \nonumber \\
\CS(H_k) &=&  \lambda_1(H_k) - \frac{2 (3k(3k -1))}{6 k^2}  =  \lambda_1(H_k) - \frac{3k -1}{k} \nonumber  \\  
\irr(H_k) &=& \sum\limits_{e \in E(H_k)}\left| d_{H_k}(u)-d_{H_k}(v)\right|
                 =    \sum\limits_{e \in E_1(H_k)}\left| d_{H_k}(u)- d_{H_k}(v)\right| = 12 (k-1).  \nonumber \\
\irr_t(H_k) &=&\frac{1}{2}\,\sum\limits_{u,v \in V(H_k)} \left| d_{H_k}(u)-d_{H_k}(v)\right|  
                    =   \frac{1}{2}\,\sum\limits_{u \in V_1(H_k), v \in V_2(H_k)} \left| d_{H_k}(u)-d_{H_k}(v)\right|  \nonumber \\
                 &=& 36\, k^2\, (k-1).       \nonumber
\eeq
\end{proof}  

A comparison between the variance and Collatz-Sinogowitz of Circumcoronene series of benzenoid $H_k$  for different values of $k$ is given in Table~\ref{comparVarCollCircum}. %
%
%
\begin{figure}[h!]
\begin{center}
\begin{tabular}{c}
 \includegraphics[scale=1.0]{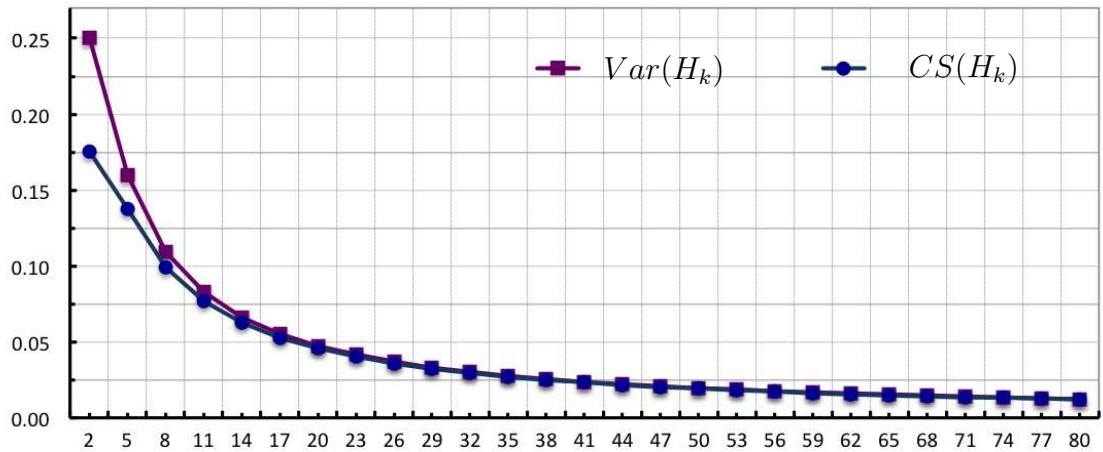} \\
\end{tabular}
\caption{Comparison between $\Var$ and $\CS$ of Circumcoronene series of benzenoid $H_k$ for $k = 2, 5, \cdots, 80$.}
\label{comparVarCollCircum}
\end{center}
\end{figure}
%
%
\subsection{Mycielski's construction $M(C_n)$ and $M(P_n)$}
%
%

The {\it Mycielski's construction} of a simple graph $G$~\cite{West}
 produces a simple graph $M(G)$ containing $G$.
Start with $G$ having vertex set $\{v_1, v_2, \cdots, v_n\}$, add
vertices $U = \{u_1, u_2, \cdots, u_n\}$ and one more vertex $w$.
Add edges to make $u_i$ adjacent to all $N_G(v_{i})$ and finally
let $N(w) = U$. One iteration of Mycielski's construction from the
graph $C_8$ and $P_8$, where $C_n$, and $P_n$ are cycle and path
of length $n$ respectively, yields the graph shown in
Figure~\ref{Mycielski}.

\begin{figure}[htbp!]
\begin{center}
\includegraphics[scale=0.8]{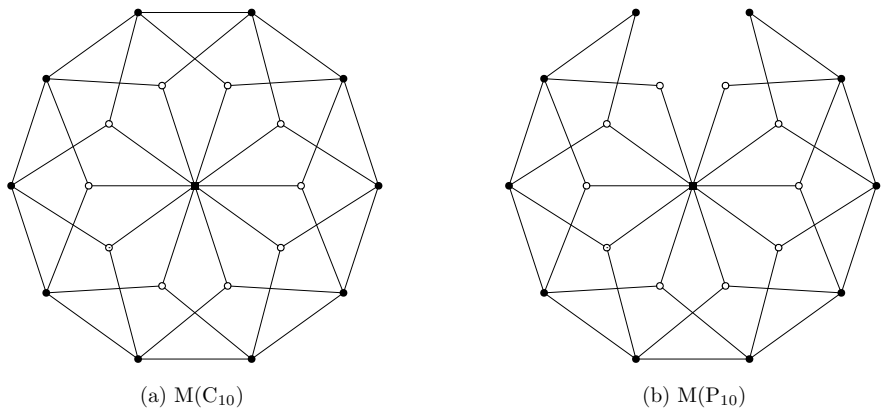}
\caption{$(a)$ Mycielski's graph $M(C_8)$,  $(b)$ Mycielski's graph $M(P_8)$.}
\label{Mycielski}
\end{center}
\end{figure}

\begin{te} \label{te-dendrimer}
Let $M(C_n)$ and $M(P_n)$  be Mycielski's graph of cycle and path graphs with $n$ vertices.
Then, 
\beq
\begin{aligned}[t]
&{ \Var}(M(C_n)) =   \frac{n (2\, n^2\, - 13 n + 25)}{(2n+1)^2}, 
&&\qquad {\CS}(M(C_n))=   \lambda_1(M(C_n)) - \frac{8n}{2n+1},  \nonumber \\ 
 & \irr(M(C_n)) =  n (n-1), 
 && \qquad  \irr_t(M(C_n)) =    n(3n-7),  \nonumber  \\ 
 &{ \Var}(M(P_n)) =   \frac{(n-2)(2 n^2 - 9 n + 35)}{(1 + 2 n)^2}, 
&& \qquad {\CS}(M(P_n))= \lambda_1(M(P_n)) - \frac{2(4n-3)}{2n+1},  \nonumber \\ 
 & \irr(M(P_n)) =  n^2 - n + 6,
 && \qquad  \irr_t(M(P_n)) =  (n-2)(3n + 7).  \nonumber  
 \end{aligned}
\eeq
\end{te}

Straightforward calculations gives that $|V(M(C_n))| = 2n +1$,
$|E(M(C_n))| = 4 n$. Hence, 
%
\beq
\Var(M(C_n))&=& \frac{1}{2n+1} \sum\limits_{v \in V(M(C_n))} d^2_{M(C_n)}(v)  -
                            \frac{1}{(2n+1)^2} \left( \sum\limits_{v \in V(M(C_n))} d_{M(C_n)}(v) \right)^2       \nonumber \\
           &=& \frac{1}{2n+1}\left(\sum\limits_{v \in V(M(C_n))} d^2_{M(C_n)}(v) +\sum\limits_{u\in U} d^2_{M(C_n)}(u)+n^2\right) \nonumber \\
            & & -\frac{1}{(2n+1)^2}\left(\sum\limits_{v \in V(M(C_n))} d_{M(C_n)}(v) +\sum\limits_{u\in U} d_{M(C_n)}(u)+n\right)^2 \nonumber \\
             &=&   \frac{4^2 \, n + 3^2 \, n + n^2}{2n+1} - \left(\frac{4n + 3n + n}{2n+1}\right)^2
                = \,\frac{n (2\, n^2\, - 13 n + 25)}{(2n+1)^2}.         \nonumber \\
\CS(M(C_n)) &=& \lambda_1(M(C_n)) - \frac{8n}{2n+1} \nonumber
\eeq
\beq
\irr(M(C_n)) &=& \sum\limits_{uv \in E(M(C_n))}\left| d_{M(C_n)}(u)-d_{M(C_n)}(v)\right|
                     =  n (n-3) + 2n     = n (n-1),     \nonumber
\eeq
\beq
\irr_t(M(C_n)) &=& \,\frac{1}{2}\,\sum\limits_{u, \, v \in V(M(C_n))}       
                                   \left| d_{M(C_n)}(u)-d_{M(C_n)}(v)\right|
                       = n(n-4) + n(n-3) + n^2             \nonumber \\
                     &=& n(3n-7).                                   \nonumber
\eeq

A direct calculations gives that $|V(M(P_n))| = 2n +1$,
$|E(M(P_n))| = 4 n - 3$. The four considered irregularity
measures have the following values:
\beq
\Var(M(P_n))&=& \frac{1}{2n+1} \sum\limits_{v \in V(M(P_n))} d^2_{M(P_n)}(v)  -
               \frac{1}{(2n+1)^2} \left( \sum\limits_{v \in V(M(P_n))} d_{M(P_n)}(v) \right)^2       \nonumber \\
           &=& \frac{1}{2n+1}\left(\sum\limits_{v \in V(P_n)} d^2_{M(P_n)}(v) +\sum\limits_{u\in U} d^2_{M(P_n)}(u)+n^2\right) \nonumber \\
           &  & -\frac{1}{(2n+1)^2}\left(\sum\limits_{v \in V(P_n)} d_{M(P_n)}(v) +\sum\limits_{u\in U} d_{M(P_n)}(u)+n\right)^2 \nonumber \\
            &=&  \frac{ n^2 + 25 n - 34}{2n+1} - \left(\frac{8n-6}{2n+1}\right)^2
               = \frac{(n-2)(2 n^2 - 9 n + 35)}{(2 n + 1)^2}.                   \nonumber \\
\CS(M(P_n)) &=& \lambda_1(M(P_n)) - \frac{2(4n-3)}{2n+1} \nonumber \\
\irr(M(P_n))   &=& \sum\limits_{uv \in E(M(P_n))}\left| d_{M(P_n)}(u)-d_{M(P_n)}(v)\right|  \nonumber \\
                     &=& (n-2)(n-3) + 2(n-2) + 2(n-3) + 8 + 2  
                       =    n^2 - n + 6                              \nonumber \\
\irr_t(M(P_n)) &=& \,\frac{1}{2}\,\sum\limits_{u \in V(M(P_n))} \sum\limits_{v \in V(M(P_n))}
                                   \left| d_{M(P_n)}(u)-d_{M(P_n)}(v)\right|      \nonumber \\
                      &=& (n-2)^2  + (n-2)(n-3) + (n-2)(n-4) + (n-2)(4+4 +4+2+2)                            \nonumber \\
                      &=& (n-2)(3n + 7).                                    \nonumber
\eeq
\begin{figure}[h!]
\begin{center}
\begin{tabular}{cc}
\includegraphics[scale=.62]{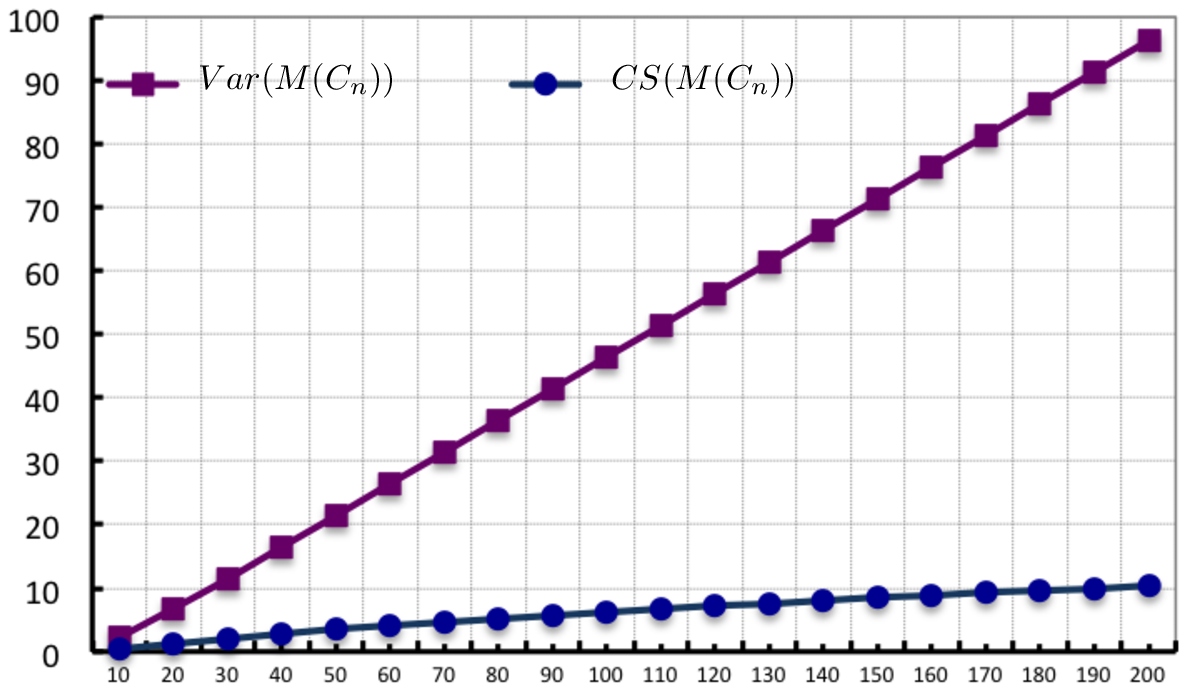} &  \includegraphics[scale=.62]{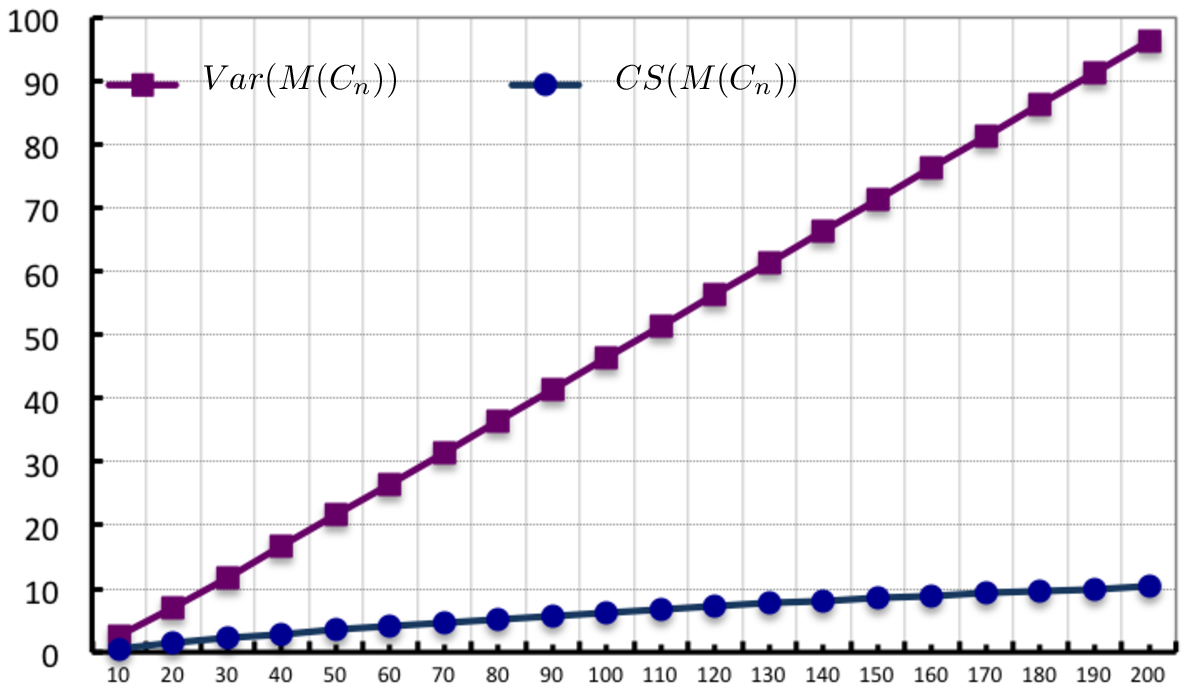}\\
(a) & (b) \\
\end{tabular}
\caption{Comparison between $\Var$ and $\CS$ of Mycielski's construction $(a)$ $M(C_n)$ and $(b)$ $M(P_n)$ for $n = 10, 20, \cdots, 200$.}
\label{comparVarCollMycielCnPn}
\end{center}
\end{figure}
\section{Concluding comments}
With the rapid development of industry, including the medical field, a great deal of new
chemical structures are being discovered and synthesized annually. This
requires to spend more on detecting the characteristics of the
many new drugs, materials and chemical compounds. Irregularity
indices may help to measure the chemical, biological and nano
properties which are widely popular in developing areas. 
In our article, in view of structure analysis and
mathematical derivation, we report the irregularity related
indices of certain molecular graphs which widely appear in
nanoscience and drug structures. 

To determine the CS index of the considered chemical structures, we
have constructed the adjacency matrix of the underlying graph
and then calculate its eigenvalues.
Since the presented chemical compounds are very well structured, with repeating rules/patterns, 
we hope that it is possible to calculate the closed-form solutions
of the CS index in those cases. This demanding task remains an open problem and could be
considered for future work.

We conclude with the following conjecture that was deduced from the
experimental part of this work.

\begin{conjecture}
Let  $G$ be a nanotube $TUC_4 C_8(S)$, $TUC_4 C_8(R)$, 
$TUHC_{6}$, $TUC_4$, $TUVC_{6}$ or circumcoronene series of benzenoid $H_k$, $k \geq 1$, 
and let  $n$ be the order of $G$. Then, 
$$
\lim_{n \to \infty} (\Var(G) - \CS(G))=0.
$$
\end{conjecture}

%
\noindent

%
\begin{appendices}
  \section{Functions written in Matlab for computing the adjacency matrix of the considered molecular structures} \label{app:foobar}
\vspace{.25cm}
\subsection{A function for computing the adjacency matrix of $TUC_4C_8(S)[p,q]$ nanotube} \label{app:TUC4C8(S)}
\vspace{.25cm}
 \begin{changemargin}{1.0cm}{0.5cm} 
 \begin{alltt}
\textcolor{gray}{01\qquad} \textcolor{keyword}{function} A = AdjMatrTUC4S4\_S(p,q)
\textcolor{gray}{02\qquad} n = 4*p*q;
\textcolor{gray}{03\qquad} \textcolor{keyword}{for} i = 1 : n
\textcolor{gray}{04\qquad}     \textcolor{keyword}{if} rem(i , 4*p) == 0
\textcolor{gray}{05\qquad}         A(i , i - 4*p + 1) = 1;  A(i - 4*p + 1 , i) = 1;
\textcolor{gray}{06\qquad}     \textcolor{keyword}{else}
\textcolor{gray}{07\qquad}         A(i , i + 1) = 1;  A(i + 1 , i) = 1;
\textcolor{gray}{08\qquad}     \textcolor{keyword}{end}
\textcolor{gray}{09\qquad}     \textcolor{keyword}{if} (rem(i , 4) == 1 | rem(i , 4) == 2) \&\& i < n - 4*p
\textcolor{gray}{10\qquad}         A(i , i + 4*p + 2) = 1;   A(i + 4*p + 2 , i) = 1;
\textcolor{gray}{11\qquad}     \textcolor{keyword}{end}
\textcolor{gray}{12\qquad} \textcolor{keyword}{end}
\textcolor{gray}{13\qquad} \textcolor{keyword}{end}
\end{alltt}
  \end{changemargin}
\vspace{0.25cm}
\subsection{A function for computing  the adjacency matrix of $TUC_4C_8(R)[p,q]$ nanotube} \label{app:TUC4C8(R)}
\vspace{.25cm}
 \begin{changemargin}{1.0cm}{0.5cm} 
\begin{alltt}
\textcolor{gray}{01\qquad} \textcolor{keyword}{function} A = AdjMatrTUC4S4\_R(p,q)
\textcolor{gray}{02\qquad} j = 3 ; k = 4; n = 4 * p * q;
\textcolor{gray}{03\qquad} \textcolor{keyword}{for} i = 1 :  n
\textcolor{gray}{04\qquad}     \textcolor{keyword}{if} rem(i , 4) == 0
\textcolor{gray}{05\qquad}         A(i , i - 3) = 1; A(i - 3 , i) = 1;
\textcolor{gray}{06\qquad}         \textcolor{keyword}{if} i <= n - 4*p 
\textcolor{gray}{07\qquad}             A(i , 4*p + i - 2) = 1;   A(4*p + i - 2 , i) = 1;
\textcolor{gray}{08\qquad}         \textcolor{keyword}{end}
\textcolor{gray}{09\qquad}     \textcolor{keyword}{else}
\textcolor{gray}{10\qquad}         A(i , i + 1) = 1;  A(i + 1 , i) = 1;
\textcolor{gray}{11\qquad}     \textcolor{keyword}{end}
\textcolor{gray}{12\qquad}     \textcolor{keyword}{while} j  <  n
\textcolor{gray}{13\qquad}         \textcolor{keyword}{if} j  ==  k*p  -  1
\textcolor{gray}{14\qquad}             A(j  ,  j - 4*p + 2) = 1;   A(j - 4*p + 2  ,  j) = 1;
\textcolor{gray}{15\qquad}             k = k + 4;
\textcolor{gray}{16\qquad}         \textcolor{keyword}{else}
\textcolor{gray}{17\qquad}             A(j  ,  j + 2)  =  1;   A(j + 2  ,  j)  =  1;
\textcolor{gray}{18\qquad}         \textcolor{keyword}{end}
\textcolor{gray}{19\qquad}         j = j + 4;
\textcolor{gray}{20\qquad}     \textcolor{keyword}{end}
\textcolor{gray}{21\qquad} \textcolor{keyword}{end}
\end{alltt}
  \end{changemargin}
\vspace{0.25cm}
\subsection{A function for computing  the adjacency matrix of $TUC_4(m,n)$ nanotube} \label{app:TUC4}
\vspace{.25cm}
 \begin{changemargin}{1.0cm}{0.5cm} 
 \begin{alltt}
\textcolor{gray}{01\qquad} \textcolor{keyword}{function} A = AdjMatrTUC4(p,q)
\textcolor{gray}{02\qquad} A  = [];
\textcolor{gray}{03\qquad} \textcolor{keyword}{for} i = 1 :  p*(q - 1)
\textcolor{gray}{04\qquad}     A(i , i + p) = 1; A(i + p , i) = 1;
\textcolor{gray}{05\qquad} \textcolor{keyword}{end}
\textcolor{gray}{06\qquad} \textcolor{keyword}{while} i  <=  p*q  -  1
\textcolor{gray}{07\qquad}     \textcolor{keyword}{for} i = 1 :  p*q  - 1
\textcolor{gray}{08\qquad}         A(i , i + 1) = 1;  A(i + 1 , i) = 1;
\textcolor{gray}{09\qquad}         \textcolor{keyword}{if} rem(i,p) == 0
\textcolor{gray}{10\qquad}             A(i , i + 1) = 0;   A(i + 1 , i) = 0;
\textcolor{gray}{11\qquad}         \textcolor{keyword}{end}
\textcolor{gray}{12\qquad}     \textcolor{keyword}{end}
\textcolor{gray}{13\qquad}     i = i + 1;
\textcolor{gray}{14\qquad} \textcolor{keyword}{end}
\textcolor{gray}{15\qquad} \textcolor{keyword}{for} i = 1 : p : p*(q - 1) + 1
\textcolor{gray}{16\qquad}     A(i , i + p - 1) = 1;  A(i + p - 1 , i) = 1;
\textcolor{gray}{17\qquad} \textcolor{keyword}{end}
\textcolor{gray}{18\qquad} \textcolor{keyword}{end}
\end{alltt}
  \end{changemargin}
\vspace{.25cm}
\subsection{A function for computing  the adjacency matrix of of Zig-Zag $TUHC_6$ nanotube} \label{app:TUHC6}
\vspace{.25cm}
 \begin{changemargin}{1.0cm}{0.5cm} 
 \begin{alltt}
\textcolor{gray}{01\qquad} \textcolor{keyword}{function} A = AdjMatrTUHC(p,q) 
\textcolor{gray}{02\qquad} \textcolor{keyword}{for} j = 1 : 2*p*q
\textcolor{gray}{03\qquad}     \textcolor{keyword}{if} j == 2*p*q
\textcolor{gray}{04\qquad}         A(j , j - 1) = 1; A(j - 1 , j) = 1;
\textcolor{gray}{05\qquad}         A(j , j - 2*p + 1) = 1; A(j - 2*p + 1 , j) = 1;
\textcolor{gray}{06\qquad}     \textcolor{keyword}{else}
\textcolor{gray}{07\qquad}         A(j , j + 1) = 1;  A(j + 1 , j) = 1;
\textcolor{gray}{08\qquad}         \textcolor{keyword}{if} rem(j , 2*p) == 0
\textcolor{gray}{09\qquad}             A(j , j - 2*p + 1) = 1; A(j - 2*p + 1 , j) = 1;
\textcolor{gray}{10\qquad}             A(j , j + 1) = 0; A(j + 1 , j) = 0;
\textcolor{gray}{11\qquad}         \textcolor{keyword}{end}
\textcolor{gray}{12\qquad}     \textcolor{keyword}{end}   
\textcolor{gray}{13\qquad} \textcolor{keyword}{end}
\textcolor{gray}{14\qquad} \textcolor{keyword}{for} j = 1 : q - 1
\textcolor{gray}{15\qquad}    \textcolor{keyword}{if} rem(j,2) \~{}= 0
\textcolor{gray}{16\qquad}        \textcolor{keyword}{for} i = 2*p*(j-1) + 1 :  2  :  2*p*j
\textcolor{gray}{17\qquad}             A(i , i + 2*p) = 1;  A(i + 2*p , i) = 1;
\textcolor{gray}{	18\qquad}         \textcolor{keyword}{end}
\textcolor{gray}{19\qquad}     \textcolor{keyword}{else}
\textcolor{gray}{20\qquad}         \textcolor{keyword}{for} i = 2*p*(j-1) + 2  :  2  :  2*p*j
\textcolor{gray}{21\qquad}             A(i , i + 2*p) = 1;  A(i + 2*p , i) = 1;
\textcolor{gray}{22\qquad}         \textcolor{keyword}{end}
\textcolor{gray}{23\qquad}     \textcolor{keyword}{end}            
\textcolor{gray}{24\qquad} \textcolor{keyword}{end}
\textcolor{gray}{25\qquad} \textcolor{keyword}{end}
\end{alltt}
  \end{changemargin}
\vspace{.25cm}
\subsection{A function for computing  the adjacency matrix of Armchair $TUVC_6[p,q]$ nanotube} \label{app:TUVC6}
\vspace{.25cm}
 \begin{changemargin}{1.0cm}{0.5cm} 
\begin{alltt} 
\textcolor{gray}{01\qquad} \textcolor{keyword}{function} A = AdjMatrTUVC(p,q)
\textcolor{gray}{02\qquad} A = [ ];   j = 1; 
\textcolor{gray}{03\qquad} \textcolor{keyword}{for} i = 1 :  2*p*q - 2*p
\textcolor{gray}{04\qquad}     A(i , i+2*p) = 1;  A(i+2*p , i) = 1;
\textcolor{gray}{05\qquad}\textcolor{keyword}{end}
\textcolor{gray}{06\qquad} \textcolor{keyword}{while} j  <=  2*p*q - 1
\textcolor{gray}{07\qquad}     A(j , j+1) = 1;  A(j+1 , j) = 1;
\textcolor{gray}{08\qquad}     \textcolor{keyword}{if} rem(j+1 , 2*p) == 0 | rem(j+2 , 2*p) == 0
\textcolor{gray}{09\qquad}        \textcolor{keyword}{if} rem(j+2 , 4*p) == 0
\textcolor{gray}{10\qquad}            A(j+2 , j-2*p+3) = 1; A(j-2*p+3 , j+2) = 1;
\textcolor{gray}{11\qquad}         \textcolor{keyword}{end}
\textcolor{gray}{12\qquad}         j = j + 1;
\textcolor{gray}{13\qquad}     \textcolor{keyword}{end}
\textcolor{gray}{14\qquad}     j = j + 2;
\textcolor{gray}{15\qquad}  \textcolor{keyword}{end}
\textcolor{gray}{16\qquad}  \textcolor{keyword}{end}
\end{alltt}
\end{changemargin} 
\vspace{.25cm}
\subsection{A function for computing  the adjacency matrix of dendrimers(k,d)} \label{app:dendrimers}
 \vspace{.25cm}
 \begin{changemargin}{1.0cm}{0.5cm} 
\begin{alltt}
\textcolor{gray}{01\qquad} \textcolor{keyword}{function} A = AdjMatrDendrimers(k,d) 
\textcolor{gray}{02\qquad} Xsta = [];        Xend  = [];  A = [];
\textcolor{gray}{03\qquad} Xsta(1) = 2;      Xend(1) = k+1;
\textcolor{gray}{04\qquad} Xsta(2) = k + 2;  Xend(2) = k\^{}2 + 1;
\textcolor{gray}{05\qquad} A(1,2:k+1) = 1;    A(2:k+1,1) = 1;  \textcolor{comment}{%  -----   Distance = 1   -----}
\textcolor{gray}{06\qquad} t = 0;        \textcolor{comment}{%   -------------    Distance = 2    ----------------}
\textcolor{gray}{07\qquad} \textcolor{keyword}{for} i = 2 : k+1
\textcolor{gray}{08\qquad}     tt = k+2 + t*(k-1); 
\textcolor{gray}{09\qquad}     A(i,tt:tt+(k-2)) = 1; A(tt:tt+(k-2),i) = 1;
\textcolor{gray}{10\qquad}     t = t +1;
\textcolor{gray}{11\qquad} \textcolor{keyword}{end}
\textcolor{gray}{12\qquad} \textcolor{keyword}{for} j = 3 : d   \textcolor{comment}{%   -------------  Distance >= 3    ---------------}
\textcolor{gray}{13\qquad}     Xsta(j) = Xend(j-1) + 1; 
\textcolor{gray}{14\qquad}     Xend(j) = Xend(j-1) + k * (k-1)\^{}(j-1);
\textcolor{gray}{15\qquad}     i = 0;
\textcolor{gray}{16\qquad}     \textcolor{keyword}{for} k1 = Xsta(j-1) : Xend(j-1)
\textcolor{gray}{17\qquad}         k2 = Xsta(j) + (k-1)*i;
\textcolor{gray}{18\qquad}         A(k1, k2:k2+(k-2)) = 1; A(k2:k2+(k-2), k1) = 1;
\textcolor{gray}{19\qquad}         i = i + 1;
\textcolor{gray}{20\qquad}     \textcolor{keyword}{end}
\textcolor{gray}{21\qquad} \textcolor{keyword}{end}
\textcolor{gray}{22\qquad} \textcolor{keyword}{end}
\end{alltt}
 \end{changemargin}
\vspace{.25cm}
\subsection{A function for computing  the adjacency matrix of Circumcoronene(k)} \label{app:circumcoronene}
\vspace{.25cm}
 \begin{changemargin}{1.0cm}{0.5cm} 
\begin{alltt}
\textcolor{gray}{01\qquad} \textcolor{keyword}{function} A = AdjMatrCircumcoronene(k)
\textcolor{gray}{02\qquad} A = []; X = [];  Xsta = [];  Xend  = [];   
\textcolor{gray}{03\qquad} Xsta(1) = 1;   Xend(1) = 6;
\textcolor{gray}{04\qquad} A(Xsta(1) , Xend(1)) = 1; A(Xend(1), Xsta(1)) = 1;
\textcolor{gray}{05\qquad} \textcolor{keyword}{for} t = 1 : 5
\textcolor{gray}{06\qquad}     A(t, t+1) = 1; A(t+1, t) =1;
\textcolor{gray}{07\qquad} \textcolor{keyword}{end}
\textcolor{gray}{08\qquad} \textcolor{keyword}{for} j = 2 : k
\textcolor{gray}{09\qquad}     Xsta(j) = 6*(j-1)\^{}2+1;     Xend(j) = 6*j\^{}2;
\textcolor{gray}{10\qquad}     A(Xsta(j) , Xend(j)) = 1;  A(Xend(j), Xsta(j)) = 1;                                                                                                                                                                                                                                                                                                                                                                                                                                                                                                                                                                                                                                                                                                                                                                                                                                                                                                                                                                                                                                                                                                                                                                                                                                                                                                                                                                                                                                                                                                                                                                                                                                                                                                                                
\textcolor{gray}{11\qquad}     \textcolor{keyword}{for} i = Xsta(j) : Xend(j)
\textcolor{gray}{12\qquad}         A(i,i-1) = 1;     A(i-1,i) = 1;
\textcolor{gray}{13\qquad}     \textcolor{keyword}{end}
\textcolor{gray}{14\qquad} \textcolor{keyword}{end}
\textcolor{gray}{15\qquad} \textcolor{keyword}{for} j = 2 : k
\textcolor{gray}{16\qquad}     \textcolor{keyword}{if} j == 2
\textcolor{gray}{17\qquad}         \textcolor{keyword}{for} i = Xsta(j) : Xend(j)
\textcolor{gray}{18\qquad}             \textcolor{keyword}{for} t = Xsta(j-1) : Xend(j-1) - 1
\textcolor{gray}{19\qquad}                 \textcolor{keyword}{if} i == Xsta(j) + 3 * t
\textcolor{gray}{20\qquad}                     A(i,t) = 1;    A(t,i) = 1;
\textcolor{gray}{21\qquad}                 \textcolor{keyword}{end}
\textcolor{gray}{22\qquad}             \textcolor{keyword}{end}        
\textcolor{gray}{23\qquad}         \textcolor{keyword}{end}
\textcolor{gray}{24\qquad}     \textcolor{keyword}{else}
\textcolor{gray}{25\qquad}         constjj(1 :  j-1) = zeros(1, j-1);
\textcolor{gray}{26\qquad}         constj(1  :  j-1) = zeros(1, j-1); 
\textcolor{gray}{27\qquad}         constjj(1 :  j-2) = [2 : 2 : 2*j-4];
\textcolor{gray}{28\qquad}         constj(1  :  j-2) = [1 : 2 : 2*j-5];
\textcolor{gray}{29\qquad}         constjj(j-1) = 2 * j - 1;
\textcolor{gray}{30\qquad}         constj(j-1)  = 2 * (j - 2);
\textcolor{gray}{31\qquad}         \textcolor{keyword}{for} i = 1 : 6
\textcolor{gray}{32\qquad}             \textcolor{keyword}{for} t = 1 : j-1
\textcolor{gray}{33\qquad}                 kjj(t) = Xsta(j) + (2*j - 1)*(i-1) + constjj(t);
\textcolor{gray}{34\qquad}                 kj(t)  = Xsta(j-1) + (2*j - 3)*(i-1) + constj(t);
\textcolor{gray}{35\qquad}                 \textcolor{keyword}{if} kjj(t) < Xend(j) \&\& kj(t) < Xend(j-1)
\textcolor{gray}{36\qquad}                     A(kjj(t),kj(t)) = 1;    A(kj(t),kjj(t)) = 1;
\textcolor{gray}{37\qquad}                 \textcolor{keyword}{end}
\textcolor{gray}{38\qquad}             \textcolor{keyword}{end}   
\textcolor{gray}{39\qquad}        \textcolor{keyword}{end}  
\textcolor{gray}{40\qquad}     \textcolor{keyword}{end} 
\textcolor{gray}{41\qquad} \textcolor{keyword}{end}
\textcolor{gray}{42\qquad}constjj = []; constj = [];
\textcolor{gray}{43\qquad} \textcolor{keyword}{end}
\end{alltt}
 \end{changemargin}
\vspace{.25cm}
\subsection{ A function for computing  the adjacency matrix of Mycielski's graph of cycle and path graphs}\label{app:Mycielskis}
 \vspace{.25cm}
 \begin{changemargin}{1.0cm}{0.5cm} 
\begin{alltt}
\textcolor{gray}{01\qquad} \textcolor{keyword}{function} [ACn, APn] = AdjMatrMycielCnPn(n)
\textcolor{gray}{02\qquad} ACn= [];   APn= [];  ACn(1 , n) = 1;  ACn(n , 1) = 1;
\textcolor{gray}{03\qquad} \textcolor{keyword}{for} i = 1 : n - 1
\textcolor{gray}{04\qquad}     ACn(i , i+1) = 1;  ACn(i+1 , i) = 1;
\textcolor{gray}{05\qquad}     APn(i , i+1) = 1;  APn(i+1 , i) = 1;
\textcolor{gray}{06\qquad} \textcolor{keyword}{end}
\textcolor{gray}{07\qquad} [n1,m1] = size(ACn);
\textcolor{gray}{08\qquad} ACn=[ACn, ACn, zeros(n1,1);  ACn, zeros(n1,n1), ones(n1,1);  \textcolor{keyword}{...}
\textcolor{gray}{09\qquad}                                             zeros(1,n1), ones(1,n1), 0];            
\textcolor{gray}{10\qquad} APn=[APn, APn, zeros(n1,1);  APn, zeros(n1,n1), ones(n1,1); \textcolor{keyword}{...}           
\textcolor{gray}{10\qquad}                                             zeros(1,n1), ones(1,n1), 0];         
\textcolor{gray}{12\qquad} \textcolor{keyword}{end}
\end{alltt}
 \end{changemargin}
\end{appendices}
\end{document}